\documentclass[acmsmall]{acmart}

\setcopyright{none}

\usepackage{booktabs}   
\usepackage{subcaption} 
\usepackage{cancel}
\usepackage{algorithm, algorithmic}
\usepackage{amsmath}
\DeclareMathOperator*{\argmin}{arg\,min}

\usepackage{amsthm}
\usepackage{graphicx}
\usepackage{multirow}
\usepackage{multicol}
\usepackage{subcaption}
\usepackage{graphicx} 
\usepackage{caption,subcaption}
\usepackage[export]{adjustbox}
\begin{document}

\title[Short Title]{A Vertex Cut based Framework for Load Balancing and Parallelism Optimization in Multi-core Systems }         


\author{Guixiang Ma}
\orcid{nnnn-nnnn-nnnn-nnnn}             
\affiliation{
  \institution{Intel Labs}            
}
\email{guixiang.ma@intel.com}          
\author{Yao Xiao}
\affiliation{
  \institution{University of Southern California}           
}
\email{xiaoyao@usc.edu}         

\author{Theodore L. Willke}
\affiliation{
  \institution{Intel Labs}           
}
\email{ted.willke@intel.com}

\author{Nesreen K. Ahmed}
\affiliation{
  \institution{Intel Labs}           
}
\email{nesreen.k.ahmed@intel.com} 

\author{Shahin Nazarian}
\affiliation{
  \institution{University of Southern California}           
}
\email{shahin.nazarian@usc.edu} 

\author{Paul Bogdan}
\affiliation{
  \institution{University of Southern California}           
}
\email{pbogdan@usc.edu} 

\begin{abstract}
High-level applications, such as machine learning, are evolving from simple models based on multilayer perceptrons for simple image recognition to much deeper and more complex neural networks for self-driving vehicle control systems. 
The rapid increase in the consumption of memory and computational resources by these models demands the use of multi-core parallel systems to scale the execution of the complex emerging applications that depend on them.  
However, parallel programs running on high-performance computers often suffer from data communication bottlenecks, limited memory bandwidth, and synchronization overhead due to irregular critical sections. 
In this paper, we propose a framework to reduce the data communication and improve the scalability and performance of these applications in multi-core systems.
We design a vertex cut framework for partitioning LLVM IR graphs into clusters while taking into consideration the data communication and workload balance among clusters. First, we construct LLVM graphs by compiling high-level programs into LLVM IR, instrumenting code to obtain the execution order of basic blocks and the execution time for each memory operation, and analyze data dependencies in dynamic LLVM traces. Next, we formulate the problem as Weight Balanced $p$-way Vertex Cut, and propose a generic and flexible framework, wherein four different greedy algorithms are proposed for solving this problem. Lastly, we propose a memory-centric run-time mapping of the linear time complexity to map clusters generated from the vertex cut algorithms onto a multi-core platform. This mapping takes into consideration cache coherency and data communication to improve the application performance. We conclude that our best algorithm, WB-Libra, provides performance improvements of 1.56x and 1.86x over existing state-of-the-art approaches for 8 and 1024 clusters running on a multi-core platform, respectively.

\end{abstract}



\maketitle

\section{Introduction}
The massive and growing number of complex applications, such as in machine learning and big data \cite{chen2014data}, call for efficient execution to reduce the run-time overhead.
In particular, sequential programs running in single-core systems fail to provide performance improvement. Nevertheless, the parallel execution in multi-core systems is not a cure-all in that it may cause performance degradation due to load imbalance, synchronization overhead, and resource sharing. The performance of parallel execution is determined by the worst execution time among spawned threads. Therefore, load imbalance can severely impact the overall performance. On the other hand, threads compete for the underlying shared hardware resources, which increases synchronization overhead if not properly handled. 

\indent Therefore, it is crucial to study how to optimize the parallel execution of applications in multi-core systems. 
The recent work in this area has focused on fine-grained parallelism and various task-to-core mapping strategies for minimizing the execution overhead (e.g., run-time, communication cost) on multi-core systems and optimizing the execution. For example, \cite{hendrickson1995improved} propose new graph partitioning algorithms based on spectral graph theory to partition coarse-grained dataflow graphs into parallel clusters for mapping large problems onto different nodes while balancing the computational loads. \cite{devine2006parallel} develops hypergraph partitioning algorithms to better model communication requirements and represent asymmetric problems to divide computations into clusters. Moreover, some existing research \cite{murray2013naiad, fetterly2009dryadlinq, murray2011ciel} designs different systems for general-purpose distributed data-parallel computing.


\indent Despite the large number of works in this area~\cite{hendrickson2000graph,hendrickson1995improved,verbelen2013graph,hendrickson1995multi}, very few of them have considered instruction-level fine-grained parallelism, which offers a novel approach to discovering optimal parallelization degree and minimizing data communication in multi-core platforms. In this paper, we explore the instruction-level parallelism using graph partitioning techniques on the low level virtual machine (LLVM) intermediate representation (IR) \cite{lattner2004llvm} graphs and cluster-to-core mapping for optimizing the parallel execution of applications on multi-core systems. 
The recent work in~\cite{xiao2019self,xiao2017load} studies a similar problem and proposed an edge-cut approach that proposes a community detection inspired optimization framework to partition dynamic dependency graphs to automatically parallelize the execution of applications while minimizing the inter-core traffic overhead. 
Although this work has achieved better performance in the multi-core parallelism optimization compared to other baseline methods such as sequential execution framework and thread-based framework, the graph partition approach used in~\cite{xiao2017load} does not consider some important structural properties of the LLVM IR graphs, for example, the power-law degree distribution. This may lead to less-than-ideal graph partitions identified by the optimization model. 
In this paper, we consider the power-law degree distribution when designing a graph partition framework for LLVM graphs, and propose vertex-cut strategies that partition graphs for better load balancing and parallelism in multi-core systems.

Generally, there are three major challenges in designing a vertex cut framework for LLVM IR graphs:
(1) How to formulate the goal of reducing data communication and optimal balanced workloads among multiple cores into the vertex cut graph partitioning problem, (2) How to incorporate edge weights into the vertex cut optimization problem, though most of the existing vertex cut methods are designed for unweighted graphs. However, the LLVM IR graphs are naturally weighted graphs, where vertices represent instructions, edges represent dynamic data dependencies among the instructions, and edge weights represent the estimated execution time for memory operations, which are crucial for measuring the expected workloads for executing instructions, and (3) How to map the graph partitions (i.e., clusters) generated by the vertex-cut approach to system's processors at run-time.
\begin{figure}
\centering
\includegraphics[width=0.98\textwidth, height=0.2\textwidth]{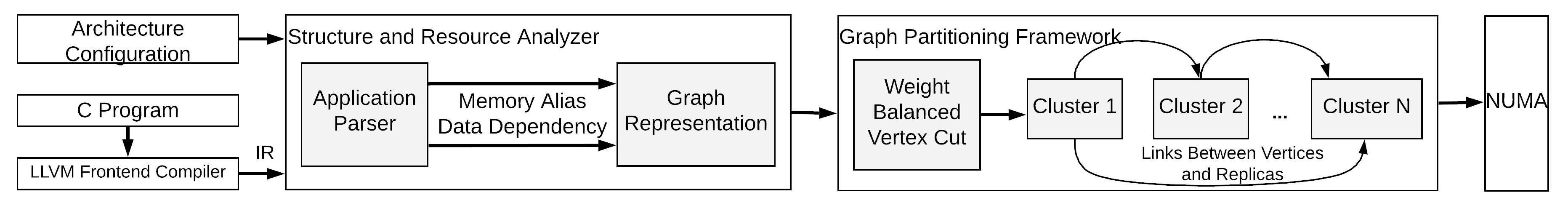}
\vspace{-3mm}
\caption{Overview of the Proposed Framework. We first pass programs into the structure and resource analyzer to construct LLVM graphs which capture spatial and temporal data communication. Next, we propose a vertex-cut based graph partitioning framework for LLVM graphs to obtain balanced clusters with minimized inter-cluster data communication. Finally, we develop a memory-centric run-time mapping to schedule clusters onto a multi-core non-uniform memory access (NUMA) platform.}
\label{fig-intro}
\vspace{-3mm}
\end{figure}

\paragraph{Contributions} To address these challenges, we propose a generic and flexible vertex cut framework on LLVM IR graphs for optimal load balancing and parallel execution of applications on multi-core systems as shown in Fig.~\ref{fig-intro}. 
The proposed framework has an advantage in incorporating power-law degree distribution into graph partitioning and can achieve extremely balanced partitions. Therefore, it is an ideal framework for balanced workloads based on graph partitioning. More specifically, we introduce and formalize a new problem, called the \emph{weight balanced $p$-way Vertex Cut}, by incorporating edge weights into the optimization of vertex cut-based graph partitioning for load balancing. In addition, we propose novel greedy algorithms for solving this problem, and introduce a memory-centric run-time mapping algorithm for mapping the graph clusters to multi-core architectures. Our contributions can be summarized as follows:  
\begin{itemize}
    \item We introduce the vertex-cut graph partition strategy to LLVM IR graphs and propose a vertex cut-based framework for partitioning LLVM graphs, which 
    reduces data communication and achieves optimally balanced workloads amongst a set of cores.
    \item We prove that the formulated optimization problems possesses submodular properties, which enables us to design a greedy algorithm for the optimization problem of the Weight Balanced $p$-way Vertex Cut with optimality guarantees. 
    \item We present a memory-centric run-time mapping algorithm for mapping the graph clusters to multiple cores. 
\end{itemize}

\paragraph{Organization} The rest of the paper is organized as follows. Section 2 describes preliminaries to help understand the paper, including edge-cut and vertex-cut graph partitioning algorithms. Section 3 provides detailed procedures for constructing LLVM graphs. Section 4 discusses the vertex-cut based graph partitioning framework and theoretical analysis. Section 5 presents the NUMA architecture and memory-centric run-time mapping. Section 6 provides simulation setup and experimental results. Finally, we discuss the related work in Section 7 and conclude the paper in Section 8.

\section{Notation \& Preliminaries}
\label{sec:prelim}
In this section, we introduce the notation and provide background for some of the fundamental concepts used throughout this paper. For a summary of notation, see Table~\ref{tab:notations}. 

\begin{table*}[]
\begin{center}
\small
\caption{Summary of Notation}
\vspace{-3mm}
\begin{tabular}{lllll}
\toprule
$G$ & & & & Input LLVM graph\\
$V$& & & & The set of vertices in a graph $G$\\
$E$& & & & The set of edges in a graph $G$\\
$W$& & & & The weight matrix for graph $G$\\
$M(e)$& & & & The set of clusters that contain edge e\\
$A(v)$& & & & The set of clusters that contain vertex v\\
$w_e$& & & & The weight of edge e\\
$\alpha$& & & & The power parameter for the power-law graphs\\
$\lambda$& & & & The edge weight imbalance factor\\
\bottomrule
\end{tabular}

\label{tab:notations}
\end{center}
\end{table*}


\paragraph{\textbf{Power-law Graphs}}
Let $G=(V,E)$ denote a graph, where $V$ is the set of vertices and $E\subseteq V\times V$ is the set of edges in $G$. Graph $G$ is a power-law graph if its degree distribution follows a power law \cite{adamic2002zipf,gonzalez2012powergraph}:
\begin{align}
    \mathbf{P}(d) \propto d^{-\alpha},
\end{align}
where $\mathbf{P}(d)$ is the probability that a vertex has a degree $d$ and $\alpha$ is a positive constant exponent. The power-law degree distribution means that most vertices in the graph have few neighbors while very few vertices have a large number of neighbors. The exponent $\alpha$ controls the "skewness" of the vertex degree distribution, where a higher $\alpha$ implies a lower ratio of edges to vertices. 
Many natural graphs have such power-law degree distributions, such as social networks. The LLVM graphs that we aim to analyze in this paper are also power-law graphs and will be introduced in Sections~\ref{sec:graph} and ~\ref{sec:framework}. Some examples of LLVM graphs are shown in Fig.~\ref{fig:llvm_graphs}. The skewed degree distributions in power-law graphs challenges graph partitioning, especially for LLVM graphs with a goal of balanced clusters and minimized data communication in parallel computing. 

\paragraph{\textbf{Edge-Cut}}
Given a graph $G = (V, E)$, an edge-cut on $G$ is a partition of $V$ into two subsets $S$ and $T$ by cutting some edges in $E$, which results in two clusters with a set of inter-cluster edges $(u,v)\in E | u\in S, v\in T$. Edge-cut based graph partitioning tasks usually have an optimization model such that after a number of edge cuts, the graph is partitioned into a certain number of clusters to satisfy the optimization requirements. Examples of edge-cut based graph partitioning problems include the widely studied max-flow min-cut problem \cite{dantzig2003max} in flow graphs and community detection in social networks \cite{bedi2016community}. In parallel computing, calculations can be considered as graphs where nodes represent a series of computations and edges represent data dependencies. Edge-cut based methods have also been studied in this area for partitioning graphs into interconnected clusters to be mapped onto parallel computers \cite{hendrickson2000graph,hendrickson1995improved,verbelen2013graph}. In this paper, we will also discuss the state-of-the-art edge-cut based methods and apply them as baseline methods for the optimal parallelism and load balancing in  multi-core systems.   

\paragraph{\textbf{Vertex-Cut}}
\begin{figure}
\centering
\begin{subfigure}[t]{0.38\textwidth}{
\centering
\includegraphics[width=0.7\textwidth,center]{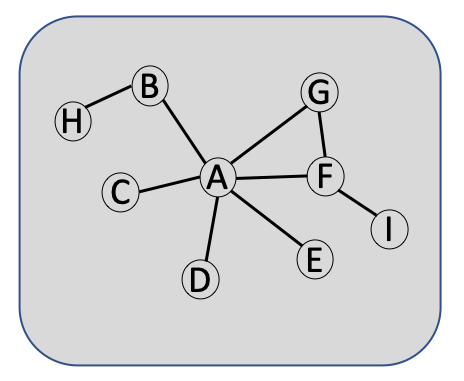}
}
\vspace{-7pt}
\caption{Sample Graph}
\label{fig:Graph}
\end{subfigure}
\hspace{-8pt}
\begin{subfigure}[t]{0.38\textwidth}{
\centering

\includegraphics[width=0.7\textwidth,center]{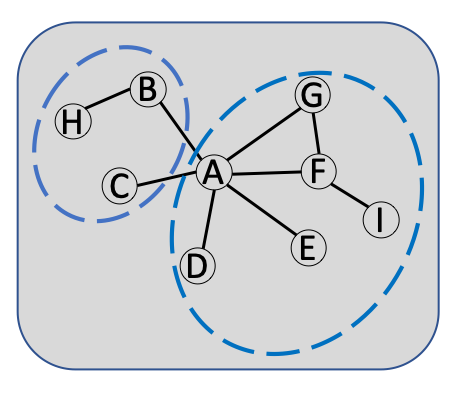}}

\caption{Edge-Cut Strategy 1}
\label{fig:EdgeCut1}
\end{subfigure}

\begin{subfigure}[t]{0.38\textwidth}{
\centering
\includegraphics[width=0.7\textwidth,center]{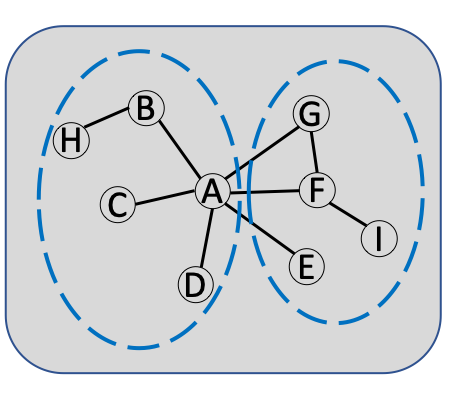}}
\caption{Edge-Cut Strategy 2}
\label{fig:EdgeCut2}
\end{subfigure}\hspace{-10pt}
\begin{subfigure}[t]{0.38\textwidth}{
\includegraphics[width=0.88\textwidth,center]{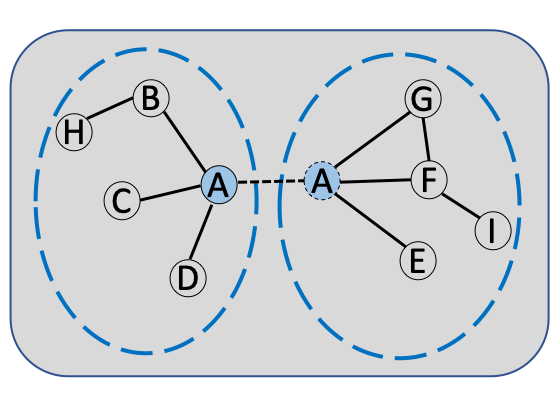}}
\caption{Vertex-Cut Strategy}
\label{fig:vertexcut}
\end{subfigure}
\vspace{-2mm}
\caption{Illustrative Example of Edge Cut vs. Vertex Cut}
\label{fig:graphcut_examples}
\end{figure}

Given a graph $G=(V,E)$, a vertex-cut on $G$ is a partition of $E$ into subsets by cutting some vertices in $E$. Whenever a vertex is cut, this vertex will be replicated and its replica along with a subset of adjacent edges are placed into a different cluster. Instead of having inter-cluster edges like an edge-cut does, vertex-cut partitions only have an inter-cluster connection between each vertex that has been cut and its replicas. Due to these characteristics, a vertex-cut strategy can offer more optimal solutions for some graph partitioning tasks compared to an edge-cut strategy. Fig.~\ref{fig:graphcut_examples} shows a scenario by illustrating the difference between edge-cut and vertex-cut on a sample graph for graph partitioning, with the goal of minimizing inter-cluster communication while balancing the workloads (i.e., the number of edges) among clusters. Since the vertex $A$ in the graph has a high degree while the other vertices have lower degrees, it is challenging for edge-cut approaches to deal with the edges associated with $A$ in order to achieve low inter-cluster communication (i.e., cross-cluster edges) and a good balance between clusters. Fig.~\ref{fig:EdgeCut1} shows an edge-cut strategy with low inter-cluster communication but a high imbalance between clusters, while the strategy in Fig.~\ref{fig:EdgeCut2} achieves a good balance but with more inter-cluster communication cost. On the other hand, the vertex-cut strategy in Fig.~\ref{fig:vertexcut} perfectly addresses the issues by cutting the vertex $A$, where the original $A$ is assigned to the cluster on the left and a replica of $A$ is assigned to the cluster on the right. The connection between $A$ and its replica is the only communication cost between the two clusters and the two clusters are well balanced. These examples demonstrate the advantage of vertex-cuts over edge-cuts on a graph with skewed node degrees, and this motivates us to propose a vertex-cut based graph partitioning framework on power-law LLVM graphs to discover the optimal execution and minimal data communication.

\section{LLVM Graph Construction}
\label{sec:graph}
\begin{figure}
\centering
\includegraphics[width=0.88\textwidth, height=0.18\textwidth]{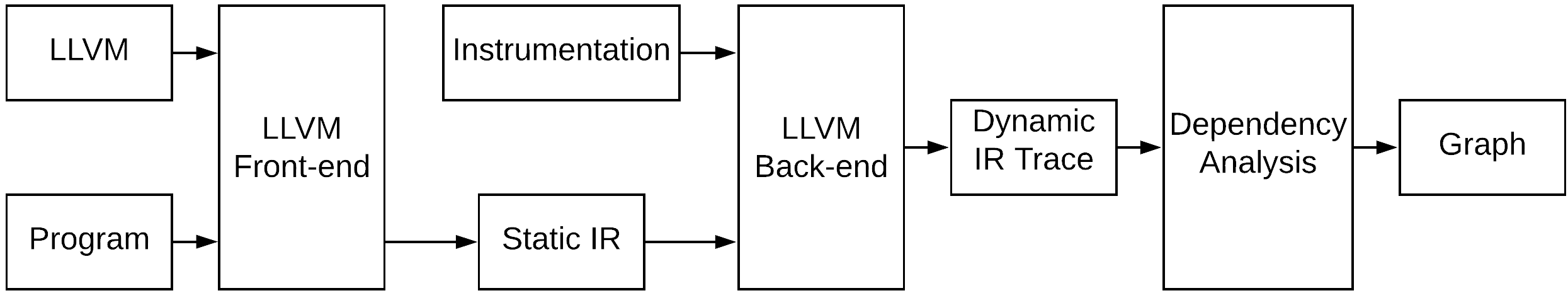} 
\caption{Workflow of LLVM Graph Construction. Each program is first compiled to static IR instructions via the LLVM front-end, which is next translated into dynamic IR trace via the LLVM back-end, combined with instrumentation to obtain information such as memory timing and the sequence of the execution order of basic blocks. Last, we perform dependency analysis to construct a graph based on the dynamic trace where nodes denote IR instructions and edges represent dependencies.}
\vskip -10pt
\label{fig:workflow}
\end{figure}
\begin{figure}[t]
\centering
\includegraphics[width=0.93\textwidth, height=0.88\textwidth]{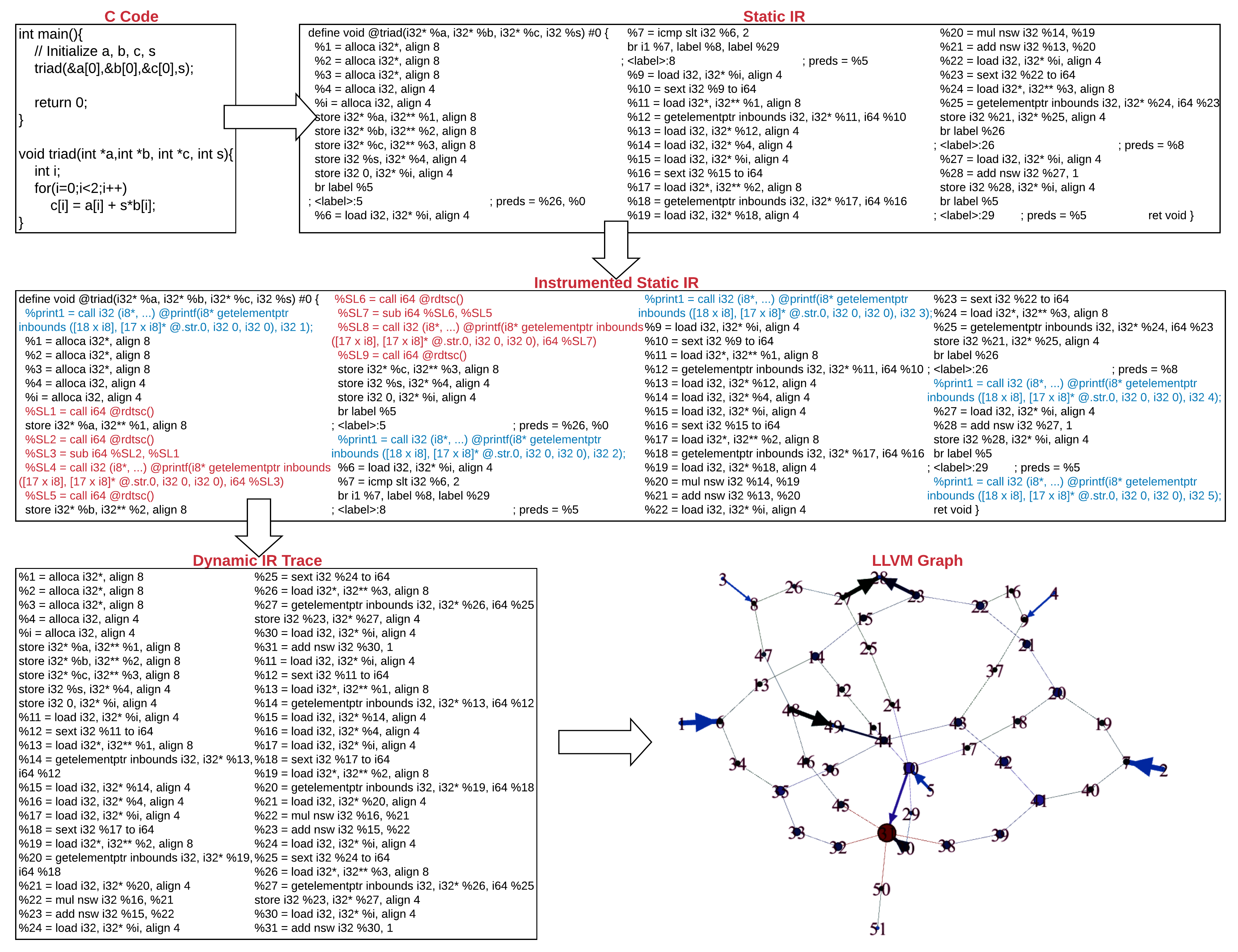} 
\caption{Example of LLVM Graph Construction. This is an example of a graph constructed from a C program followed by the workflow in Fig.~\ref{fig:workflow}. One thing to note is that in instrumented static IR, instructions in blue keep track of basic blocks whereas instructions in red measure time for memory operations. We only show partial instrumentation for memory time measurement.}
\vskip -10pt
\label{fig:example}
\end{figure}
We consider each application as an LLVM graph generated from a dynamic trace. It is a directed and acyclic graph with edge weights, where nodes represent IR instructions, edges represent dynamic data dependencies between nodes, and weights represent time for memory operations. The dataflow representation of LLVM graphs requires the advanced graph partitioning algorithms discussed later to find balanced clusters in parallel computing. In this section, we discuss the workflow of LLVM graph construction in three steps: (1) static IR generation via the LLVM front-end from an input program; (2) dynamic IR generation from static IR combined with instrumentation; (3) LLVM graph construction via dependency analysis. Before delving into the details, we introduce a graph definition to help us understand it.

\begin{definition}\label{def:llvm}
Let $G = (V,E,W)$ denote an LLVM graph, where each node $v \in V$ represents an LLVM IR instruction, $N = |V|$ is the number of nodes,  each edge $e=(u,v) \in E$ represents the data dependency among two nodes, and the corresponding edge weight $w_{uv} \in W$ characterizes the data dependency between node $u$ and node $v$ to guarantee the strict program order. 
\end{definition}

The LLVM graph in Definition~\ref{def:llvm} captures the spatial and temporal data communication, since the weight $w_{uv}$ measures the amount of time required to transfer data from node $u$ to node $v$ during memory operations. Therefore, we could measure the cost of data communication, which facilitates us to propose an optimization model to partition the LLVM graph into clusters while taking into account data transfer among clusters.

\subsection{Static IR Generation}
Instruction set architecture (ISA) dependent traces include different characteristics and constraints for a specific ISA, which cannot satisfy the fast growing hardware specialization and ever-expanding workloads. Therefore, in parallel computing, in order to have well-balanced workloads with non-trivial properties and understand the ISA-independent micro-structures, we first compile high-level programs into static LLVM IR. LLVM is a compiler engine which makes program analysis lifelong and transparent by introducing IR as a common model for analysis, transformation, and synthesis \cite{lattner2004llvm}. IR is an intermediate representation between high-level instructions such as Python/C and low-level assembly. It ignores the low-level hardware details while preserving the dataflow structure of programs, as shown in an example in Fig.~\ref{fig:example}.
\subsection{Dynamic Trace Generation}
Once the static IR code is generated, we instrument the code to obtain information such as basic blocks and memory time. First, we use a hash table to keep track of IR instructions within each basic block. For example, in Fig.~\ref{fig:example}, instructions from "\textit{\%1 = alloca i32*, align 8}" up to "\textit{br label \%5}" should be hashed into the index 1 which represents the first basic block. Second, at the beginning of each basic block, we instrument a \textit{printf} function to record the execution order of blocks. Fig.~\ref{fig:example} shows the full instrumentation of \textit{printf} statements in blue. Last, we use the time stamp counter \textit{rdtsc} and the \textit{printf} statements to measure the amount of time for each memory operation ($time=after_{mem}-before_{mem}$). Instructions in red in Fig.~\ref{fig:example} show this instrumentation for the first two memory operations. Specifically, we insert \textit{rdtsc} before and after each memory operation and calculate the difference as the amount of execution time. Once static IR is instrumented, we use the LLVM back-end to execute it and collect the execution order of basic blocks and the amount of time for each memory operation. Combined with the hash table, which can be indexed from the execution order of basic blocks, we obtain dynamic IR trace as shown in Fig.~\ref{fig:example}. 
\subsection{LLVM Graph Construction}
The dynamic IR trace captures the dataflow nature of high-level programs. In order to understand the hidden communication structure of the trace and processes that can be potentially be processed in parallel, we construct the LLVM graph by analyzing the data and memory alias dependencies. Data dependency analysis identifies source registers and destination registers for each instruction and checks if source registers of the current instruction match with destination registers of the prior ones. Alias analysis is used to determine if two pointers used in memory operations have the same address. For example, the sixth instruction "\textit{store i32* \%a, i32** \%1, align 8}" has the source register \%1 which depends on the destination register of the first instruction "\textit{\%1 = alloca i32*, align 8}". The corresponding LLVM graph manifests this dependency by inserting a directed edge from node $1$ to node $6$.

\section{Vertex-Cut Based Graph Partitioning Framework}

In this section, we introduce our vertex cut framework for partitioning LLVM graphs to optimize the parallel execution of applications in multi-core systems.  

By investigating the degree distribution of the LLVM graphs that we construct following the procedures introduced in Section~\ref{sec:graph} for some applications, we observe that these LLVM graphs are all power-law graphs, such as the examples shown in Fig.~\ref{fig:llvm_graphs} for the Dijkstra algorithm and the fast Fourier transform (FFT) algorithm. The skewed node degree distribution makes the graph partitioning a challenging task on these power-law graphs. As discussed in Section~\ref{sec:prelim}, vertex-cut has some advantages over edge-cut on graphs with skewed node degree distributions. Existing works in graph partitioning for distributed graph computing have also shown that vertex-cut methods can achieve better performance in terms of data communication and balance among the partitions than edge-cut methods for power-law graphs \cite{gonzalez2012powergraph,xie2014distributed}. In \cite{gonzalez2012powergraph}, a vertex-cut approach called PowerGraph is proposed for solving a balanced $p$-way vertex-cut problem, where the objective is to minimize the average number of vertex replicas while keeping the number of edges balanced among different clusters, so as to minimize the data communication among different clusters while balancing their workloads. In \cite{xie2014distributed}, a degree-based vertex-cut method called Libra is proposed, which has shown better performance than PowerGraph for the balanced $p$-way vertex-cut task.     
\label{sec:framework}
\begin{figure}
\centering
\begin{subfigure}[t]{0.45\textwidth}{
\label{fig:dijkstra}
\includegraphics[width=\textwidth]{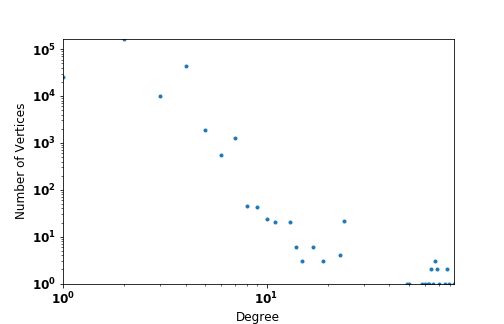}
}
\vspace{-5mm}
\caption{Dijkstra}
\end{subfigure}
\begin{subfigure}[t]{0.45\textwidth}{
\label{fig:fft}
\includegraphics[width=\textwidth]{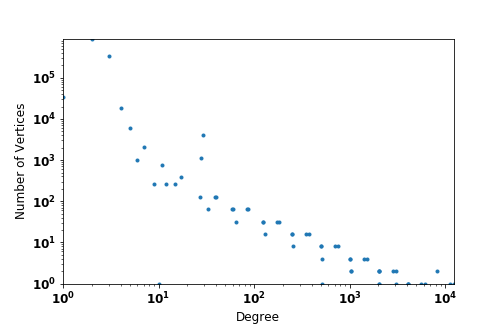}}
\par\bigskip
\vspace{-3mm}
\caption{FFT}
\end{subfigure}
\vspace{-4mm}
\caption{Examples of LLVM Graphs}
\label{fig:llvm_graphs}
\end{figure}
Although these existing vertex-cut methods have been shown effective in graph partitioning for distributed graph computing, they are designed for unweighted graphs, where the goal of a balanced cut is to keep the number of edges balanced on each cluster. However, the LLVM graphs are weighted graphs, where the weights represent the estimated execution time for memory operations, and the goal of a balanced cut is to keep the sum of edge weights in different clusters balanced. Therefore, the existing unweighted vertex cut methods cannot achieve satisfactory performance in the graph partition task on LLVM graphs. In this paper, we formulate the Weight Balanced $p$-way Vertex Cut as a new problem and propose strategies for solving this problem.        

\subsection{Weight Balanced $p$-way Vertex Cut}
\label{sec:wb-vertexcut}
Given an LLVM IR graph $G = (V,E,W)$, our goal is to reduce data communication among different cores (i.e., partitions/clusters) while achieving optimal balanced workloads (i.e., edge weights). We formalize the objective of the weight balanced $p$-way vertex-cut by assigning each edge $e \in E$ to a cluster $M(e) \in \{1,\cdots,p\}$. Each vertex then spans the set of clusters $A(v) \subseteq \{1,\cdots,p\}$ that contain its adjacent edges. We define the objective as: 
\begin{align}
    \min_A \frac{1}{|V|} \sum_{v\in V}|A(v)| \label{eq:optimization}
    \\
\text{s.t.} ~~~{\max_m\sum_{e\in E, M(e)=m} w_e, <  \lambda \frac{w_{avg}|E|}{p}}
\label{eq:constraint}
\end{align}
\noindent where $w_e$ is the weight of edge $e$, $w_{avg}$ is the average edge weight of graph $G$, and the imbalance factor $\lambda \geq 1$ is a small constant.

As previously discussed, the balanced $p$-way vertex cut for unweighted graphs has been studied in some works in the literature \cite{gonzalez2012powergraph,xie2014distributed}. \cite{gonzalez2012powergraph} introduces \textit{PowerGraph} and \cite{xie2014distributed} proposes \textit{Libra}, both of which are state-of-the-art approaches for the vertex-cut task on unweighted graphs. In \textit{PowerGraph} \cite{gonzalez2012powergraph}, randomized vertex cut strategy is first analyzed, based on which a greedy solution is proposed for the edge-placement process of the vertex-cuts. In \cite{xie2014distributed}, a degree-based approach, called \textit{Libra}, is proposed for vertex-cut graph partition, which is based on \textit{PowerGraph} but further distinguishes the higher-degree and lower-degree vertices during an edge placement to achieve better performance. Inspired by \textit{PowerGraph} and \textit{Libra}, in the following sections, we will first perform theoretical analysis on the random vertex cut solution for the proposed weighted balanced vertex cut problem and then provide greedy algorithms for this vertex cut task.   

\subsection{Theoretical Analysis}
\subsubsection{Random Weighted Vertex Cut}
\label{sec:replication}
A simple way to perform vertex-cuts is to randomly assign edges to clusters. Based on \cite{gonzalez2012powergraph}, we derive the expected normalized replication factor (Eq.(~\ref{eq:optimization})) in random weighted vertex cut for the weight balanced $p$-way vertex cut task. 

According to linearity of expectation, we have:
\begin{align}
    \mathbb{E}[\frac{1}{|V|}\sum_{v\in V}|A(v)|] = \frac{1}{|V|}\sum_{v\in V}\mathbb{E}[|A(v)|]
\end{align}
\noindent where $\mathbb{E}[|A(v)|]$ is the expected replication number of a single vertex $v$. 

Assume vertex $v$ has a degree $\mathbf{D}[v]$, then the expected replication of $v$ can be computed by the process of assigning the $\mathbf{D}[v]$ edges that are adjacent to $v$. Let $X_i$ denote the event that vertex $v$ has at least one of its edges on cluster $i$, then the expectation $\mathbb{E}[X_i]$ is:
\begin{align}
    \begin{split}
      \mathbb{E}[X_i] & =  1 - \mathbf{P}(\text{$v$ has no edges on cluster $i$})\\
      & = 1 - (1 - \frac{1}{p})^{\mathbf{D}[v]}
    \end{split}
\end{align}

Then, the expected replication factor for vertex $v$ is: 
\begin{align}
    \mathbb{E}[|A(v)|] = \sum_{i=1}^{p}\mathbb{E}[X_i] = p (1 - (1 - \frac{1}{p})^{\mathbf{D}[v]})
\label{eq:E_Av}
\end{align}

In the power-law graph, $\mathbf{D}[v]$ can be treated as a \textit{Zipf} random variable, therefore Eq.(\ref{eq:E_Av}) can be further written as:
\begin{align}
    \mathbb{E}[|A(v)|] = p(1 - \mathbb{E}[(\frac{(p-1)}{p})^{\mathbf{D}[v]}])
\end{align}

Then:
\begin{align}
    \mathbb{E}[\frac{1}{|V|}\sum_{v\in V}|A(v)|] = \frac{p}{|V|}\sum_{v\in V} (1 - \mathbb{E}[(\frac{(p-1)}{p})^{\mathbf{D}[v]}])
\label{eq:seven}
\end{align}

In the power-law graph $G$, the probability of a vertex degree being $d$ is $\mathbf{P}(d) = d^{-\alpha}/\mathbf{h}_{|V|}(\alpha)$, where $\mathbf{h}_{|V|}(\alpha)=\sum_{d=1}^{|V|-1}d^{-\alpha}$ is the normalizing constant of the power-law \textit{Zipf} distribution. Then,
\begin{align}
    \mathbb{E}[(\frac{(p-1)}{p})^{\mathbf{D}[v]}] = \frac{1}{\mathbf{h}_{|V|}(\alpha)}\sum_{d=1}^{|V|-1}(1-\frac{1}{p})^d d^{-\alpha}
\label{eq:Ep}
\end{align}
By plugging Eq.(\ref{eq:Ep}) into Eq.(\ref{eq:seven}), we have:
\begin{align}
    \mathbb{E}[\frac{1}{|V|}\sum_{v\in V}|A(v)|] = p - \frac{p}{\mathbf{h}_{|V|}(\alpha)}\sum_{d=1}^{|V|-1}(\frac{p-1}{p})^d d^{-\alpha}
    \label{eq:expected_rep}
\end{align}

We can improve the randomly weighted vertex cut with greedy strategies which assign next edge onto the cluster that minimizes the conditional expected replication factor. But before we discuss greedy based algorithms for the weighted balanced vertex cut, we first prove that the objective function in Eq. (\ref{eq:optimization}) is submodular and a greedy algorithm can provide bounded optimality.

\subsubsection{Submodularity of the Objective Function}
\begin{theorem}\label{thm:k-balance}
The optimization problem introduced in Section~\ref{sec:wb-vertexcut} is NP-hard.
\end{theorem}
\begin{proof}
K-balanced graph partitioning \cite{andreev2006balanced} divides a graph into $k$ equal sized clusters while minimizing the capacity of edges cut, which is NP-hard. It reduces to the optimization problem by having a unit weight for each edge in a graph to be cut.
\end{proof}

\begin{theorem}\label{thm:submodular}
The objective function in the Eq. (\ref{eq:optimization}) is submodular.
\end{theorem}
\begin{proof}
Given an LLVM IR graph $G=(V,E,W)$, define two assignment sets $X,Y=\{A(v_1), A(v_2), \\...,A(v_{|V|})\}\subseteq \Omega$ where for any node $v$, $A(v) \subseteq \{1,\cdots,p\}$ and $\Omega$ is the solution space of the problem. We define $f(X)$ as the objective function defined in Eq. (\ref{eq:optimization}).\\
\indent If $X\cap Y = \emptyset$, then
\begin{align}
    \begin{split}
      f(X\cap Y)+f(X\cup Y) & = \frac{1}{|V|}\sum_{v\in V}|X(v)+Y(v)|+\cancelto{0}{f(\emptyset)}\\
      & = \frac{1}{|V|}\sum_{v\in V}|X(v)|+\frac{1}{|V|}\sum_{v\in V}|Y(v)|\\
      & = f(X)+f(Y)
    \end{split}
\label{eq:sumodular1}
\end{align}
\indent If $X\cap Y = S_c$ where $S_c$ is a set of the common elements, then
\begin{align}
    \begin{split}
      f(X\cap Y)+f(X\cup Y) & = \frac{1}{|V|}\sum_{v\in V}|X(v)+Y(v)-S_c(v)|+\frac{1}{|V|}\sum_{v\in V}|S_c(v)|\\
      & = \frac{1}{|V|}\sum_{v\in V}\{|X(v)|+|Y(v)|-|S_c(v)|+|S_c(v)|\}\\
      & = \frac{1}{|V|}\sum_{v\in V}|X(v)|+\frac{1}{|V|}\sum_{v\in V}|Y(v)|\\
      & = f(X)+f(Y)
    \end{split}
\label{eq:submodular2}
\end{align}
\indent Therefore, by combining Eq. (\ref{eq:sumodular1}) and Eq. (\ref{eq:submodular2}), we can infer that the objective function is submodular because for any two sets $X, Y \subseteq \Omega$, $f(X) + f(Y) = f(X \cap Y) + f(X\cup Y)$. 
\end{proof}
\begin{theorem}\label{thm:monotonic}
The objective function in the Eq. (\ref{eq:optimization}) is monotonic.
\end{theorem}
\begin{proof}
\indent Given an LLVM IR graph $G=(V,E,W)$, we define an assignment set $A=\{v_1, v_2, ...,v_{|V|}\}\subseteq \Omega$ and an arbitrary assignment $v_k \nsubseteq A$.
\begin{align}
    \begin{split}
      f(A\cup v_k) & = \frac{1}{|V|}\sum_{v\in V}|A'(v)|\\
      & = \frac{1}{|V|}\sum_{v\in V}\{|A(v)|+|v_k(v)|\}\\
      & = \frac{1}{|V|}\sum_{v\in V}|A(v)|+\frac{1}{|V|}\sum_{v\in V}|v_k(v)|\\
      & = f(A)+f(v_k)
    \end{split}
\end{align}
\indent Therefore, $f(A\cup v_k)-f(A)\geqslant 0$.
\end{proof}

\begin{theorem}\label{thm:montonic}
Given a monotonic submodular function $f$, the greedy maximization algorithm\footnote{We can easily convert minimization to maximization in this problem by adding a negative sign to the function.} returns \begin{equation}
    f(A_{greedy}) \geqslant (1-\frac{1}{e})\max_{|A|<K}f(A)
\end{equation}
where $K$ is the maximum number of possible assignments. Therefore, even though the optimization problem is NP-hard, algorithm 1 is designed to find an assignment which provides a $(1-1/e)$ approximation of the optimal value of $A$.
\end{theorem}

\noindent
The proof of Theorem~\ref{thm:montonic} follows from~\cite{krause2014submodular}.
 
\subsection{Greedy Algorithms for Weight Balanced Vertex Cut}
\label{sec:algorithms}
\begin{algorithm}
\caption{Weight Balanced Libra: A Greedy Algorithm for Vertex Cut Graph Partitioning}
 \begin{algorithmic}[1]
    \STATE \textbf{Input}: Edge set $E$; edge weight matrix $W$; vertex set $V$; a set of clusters $C=\{1,2,\cdots,p\}$; $\lambda$.
    \STATE \textbf{Output}: The assignment $M(e)\in \{1,2,\cdots,p\}$ of each edge $e$.
    \STATE Count the degree $d_i$ for each vertex $v_i$, $\forall i \in \{1,2,\cdots,|E|\}$
    \STATE Compute the cluster weight sum bound $b=\lambda \frac{\sum_{e\in E}w_e}{p}$
    \FOR{each $e = (v_i, v_j)\in E$}
        \IF{$A(v_i)=\emptyset$ and $A(v_j)=\emptyset$}
            \STATE $m = leastloaded(C)$
        \ELSIF{$A(v_i)\neq \emptyset \land A(v_j)=\emptyset$}
            \STATE $m = leastloaded(A(v_i))$
            \IF{$load(m)\geq b$}
                \STATE $m = leastloaded(C)$
            \ENDIF
        \ELSIF{$A(v_i)= \emptyset \land A(v_j)\neq \emptyset$}
            \STATE $m = leastloaded(A(v_j))$
            \IF{$load(m)\geq b$}
                \STATE $m = leastloaded(C)$
            \ENDIF
        \ELSIF{$A(v_i)\cap A(v_j) \neq \emptyset$}
            \STATE $m = leastloaded(A(v_i)\cap A(v_j))$
            \IF{$load(m)\geq b$}
                \STATE $m = leastloaded(A(v_i)\cup A(v_j))$
                \IF{$load(m)\geq b$}
                    \STATE $m = leastloaded(C)$
                \ENDIF
            \ENDIF
        \ELSE
            \STATE $s = arg \min_l{\{d_l|l\in \{i,j\}})$
            \STATE $t = \{v_i,v_j\}-\{s\}$
            \STATE $m = leastloaded(A(s))$
            \IF{$load(m)\geq b$}
                \STATE $m = leastloaded(A(t))$
                \IF{$load(m)\geq b$}
                    \STATE $m = leastloaded(C)$
                \ENDIF
            \ENDIF
        \ENDIF
    \STATE $M(e)\leftarrow m $; $A(v_i)\leftarrow m $; $A(v_j)\leftarrow m$
    \ENDFOR
\end{algorithmic}
\label{algorithm:wblibra}
\end{algorithm}

To solve the vertex cut optimization problem defined in Eq. (\ref{eq:optimization}) via a greedy approach, we consider the task of placing the $(i+1)$-th edge after having placed the previous $i$ edges. We define the objective based on the conditional expectation, as shown below.

\begin{align}
    \argmin_k\mathbb{E} \Big[\sum_{v\in V} |A(v)| \;\Big|\; A_i,A(e_{i+1}) = k \Big]
\end{align}

\noindent where $A_i$ is the assignment for the previous $i$ edges.  

In the following paragraphs, we propose four different greedy solutions for the edge placement of the weight balanced vertex-cut. We call them Weighted PowerGraph, Weight Balanced PowerGraph, Weighted Libra, and Weight Balanced Libra.
\paragraph{\textbf{Weighted PowerGraph}} The \textit{PowerGraph} approach is proposed in \cite{gonzalez2012powergraph} for unweighted vertex cuts, which assigns edges to clusters while balancing the number of edges assigned to each cluster. Inspired by the greedy edge placement in \textit{PowerGraph} and based on the objective of the weighted vertex cut defined in Eq. (\ref{eq:optimization}), we define the edge placement rules for our Weighted PowerGraph greedy algorithm as follows.
For an edge $(u,v)$,
\begin{itemize}
    \item Case 1: If $A(u)\cap A(v)\neq \emptyset$, then assign the edge to the least loaded cluster in $A(u)\cap A(v)$, where the workload of each cluster refers to the total weights of all the edges assigned to the cluster. 
    \item Case 2: If $A(u)\cap A(v) = \emptyset$ and $A(u) \neq \emptyset, A(v)\neq \emptyset$, then assign edge $(u,v)$ to the least loaded cluster in $A(l)$, where $l$ is the vertex from ${u,v}$ that has more unassigned edges. 
    \item Case 3: If one of $A(u)$ and $A(v)$ is not empty, then assign the edge $(u,v)$ to the least loaded cluster in the non-empty set (i.e., $A(u) \cup A(v)$).
    \item Case 4: If $A(u) = \emptyset$ and $A(v) = \emptyset$, then assign $(u,v)$ to the least loaded one of the $p$ clusters. 
\end{itemize}

\paragraph{\textbf{Weighted Libra}}
Due to the power-law degree distribution in LLVM graphs, the edge weights associated with high-degree vertices tend to accumulate in a single cluster if these vertices are not cut and spanned over multiple clusters, which can lead to workload imbalance. Moreover, cutting the higher-degree vertices tends to save more communication cost between clusters compared to cutting lower-degree vertices. The Libra unweighted vertex cut approach in \cite{xie2014distributed} first proposes a degree-based hashing strategy to address such an issue for cutting power-law graphs, where the higher-degree vertex associated with an edge will be cut with priority if a vertex has to be cut in order to place this edge. Inspired by the degree-based strategy in Libra, we exploit the degree property of vertices during edge placement. Based on Weighted PowerGraph and this degree-based rule, we propose a greedy algorithm called Weighted Libra, which has the following edge placement rules: For an edge $(u,v)$,
\begin{itemize}
    \item Case 1: If $A(u)\cap A(v)\neq \emptyset$, then assign the edge to the least loaded cluster in $A(u)\cap A(v)$, where the workload of each cluster refers to the total weights of all the edges assigned to the cluster.  
    \item Case 2: If $A(u)\cap A(v) = \emptyset$ and $A(u) \neq \emptyset, A(v)\neq \emptyset$, then assign edge $(u,v)$ to the least loaded cluster in $A(l)$, where $l$ is either $u$ or $v$ whichever has the lower degree. 
    \item Case 3: If one of $A(u)$ and $A(v)$ is not empty, then assign the edge $(u,v)$ to the least loaded machine in the non-empty set (i.e., $A(u) \cup A(v)$).
    \item Case 4: If $A(u) = \emptyset$ and $A(v) = \emptyset$, then assign $(u,v)$ to the least loaded one of the $p$ clusters. 
\end{itemize}

According to the edge placement rules of Weighted PowerGraph and Weighted Libra, the load balancing among clusters is considered by assigning edges to the least loaded cluster under each case. However, this strategy cannot guarantee the overall balance of the workload (i.e., total edge weights) among different clusters or permit control of the emphasis to put on the balance constraint. To address this issue and further improve load balancing, we incorporate an explicit constraint on the balance of edge weights among clusters into the greedy edge placement rules of the Weighted PowerGraph and Weighted Libra, and have two new greedy algorithms: \textit{\textbf{Weight Balanced PowerGraph}} and \textit{\textbf{Weight Balanced Libra}}. Specifically, we incorporate the constraint on the edge weight balance, which is formulated in Eq. (\ref{eq:constraint}), into the greedy edge placement rules of the Weighted PowerGraph and Weighted Libra. For cases 1-3 in both algorithms, before placing an edge to the target cluster, we first check if the current sum of edge weights in a target cluster is within the bound given by $\lambda \frac{w_{avg}|E|}{p}$, where $\lambda \geq 1$ is a constant. If it is, then we place the edge into this cluster.  Otherwise, we search another cluster from the remaining set that satisfies this condition as the target cluster for the placement. By setting different values to $\lambda$, we can allow different amounts of emphasis on the workload balance. To illustrate the overall workflow of these greedy algorithms, we summarize the Weighted Balanced Libra greedy algorithm in Algorithm~\ref{algorithm:wblibra} as an example. 

\subsection{Discussions}
\label{sec:discussions}
\paragraph{\textbf{Time Complexity}}
According to the workflow of the Weighted Balanced Libra algorithm as shown in Algorithm~\ref{algorithm:wblibra}, given a graph $G=(V,E,W)$, for each edge $e$ in $E$, the algorithm retrieves the cluster with the least load (i.e., total edge weights) either from the entire cluster set $C$ or from a subset of $C$. For the former case (line 7, 11, 16, 23, or 33 in Algorithm~\ref{algorithm:wblibra}), it takes $O(|C|)$ time, and for the latter case, it takes $O(|C_1|)$ time ($|C_1|\leq |C|$) if the balance constraint is satisfied (line 9, 14, 19, 29), and otherwise it takes $O(|C_1| + |C_2|), |C_2|\leq |C|$ (line 21, 31), or $O(|C_1| + |C|)$ (line 11, 16), or $O(|C_1| + |C_2| + |C|)$ (line 23, 33). Note that line 27 in the algorithm takes $O(1)$. So in the worst case, the algorithm takes $O(3|C|)+O(1)=O(|C|)$ for placing one edge. Therefore, the overall time complexity of Weighted Balanced Libra algorithm is $O(|E|\times|C|)$. Based on the edge placement rules introduced in Section~\ref{sec:algorithms}, this time complexity applies to the three other algorithms as well, although the Weighted Libra and Weighted PowerGraph may take relatively less time in practice compared to Weight Balanced PowerGraph and Weight Balanced Libra, since they do not have the weight balanced constraint. Therefore, they have the time complexity of $O(|C|)$ or $O(|C_1|)$ discussed above for placing one edge. 

\paragraph{\textbf{Edge Weight Imbalance}}
Besides the replication factor discussed in Section~\ref{sec:replication} as a goal of the optimization model for the weight balanced vertex cut problem, the edge weight balance among different clusters is another key metric for evaluating the performance, which determines the load balance. As we discussed above in Section~\ref{sec:algorithms}, the degree-based hashing strategy introduced by Libra tends to have more balanced cut results as the higher-degree vertices have a higher priority to be cut than the lower-degree vertices. This statement has also been proved theoretically by \cite{xie2014distributed}, which shows that Libra can achieve a lower edge imbalance than PowerGraph. By incorporating this degree-based vertex cut rule, our proposed Weighted Libra algorithm is expected to achieve a better load balancing (i.e., a lower edge weight imbalance) than the Weighted PowerGraph algorithm. Furthermore, the proposed Weight Balanced PowerGraph and Weight Balanced Libra allows for a further improvement in the load balancing via incorporating a constraint for the edge weight imbalance by the given bound $\lambda \frac{w_{avg}|E|}{p} (\lambda \geq 1)$. If we set $\lambda = 1$, these two algorithms can guarantee near-ideal balanced vertex cut results, with an edge weight imbalance $1+\epsilon$, where $\epsilon$ is a small non-negative constant. 


\section{Architecture and Mapping}
In this section, we discuss the non-uniform memory access (NUMA) architecture used in the evaluation and propose a memory-centric mapping scheme for mapping the clusters obtained via the vertex-cuts onto processors in the NUMA architecture.
\subsection{NUMA Architecture}
\begin{figure}
\centering
\includegraphics[width=0.80\textwidth, height=0.27\textwidth]{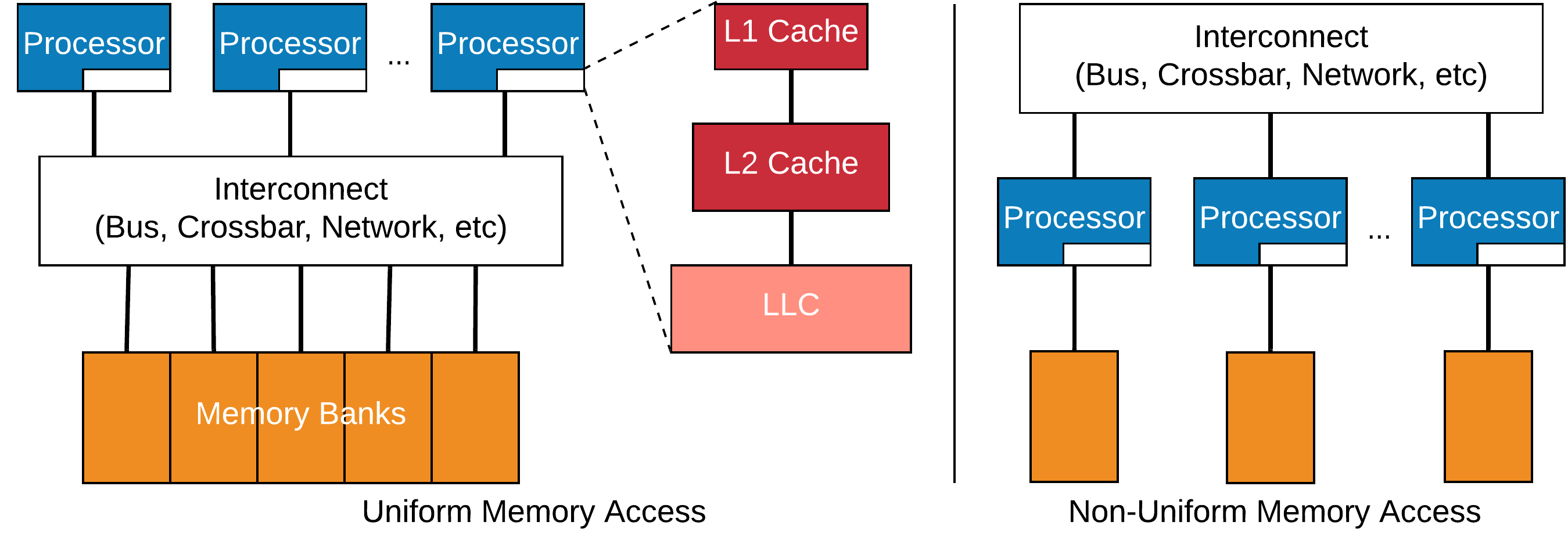} 
\vskip -5pt
\caption{UMA and NUMA Architectures. UMA is a shared memory architecture with uniform memory access whereas NUMA enables fast memory access for a processor to its local physical memory and slow memory access to the rest of memory banks.}
\vskip -10pt
\label{fig-uma}
\end{figure}
\indent Uniform memory access (UMA) is a shared memory architecture for processors running in parallel as shown in Figure \ref{fig-uma}. It develops a unified vision of the shared physical memory, meaning that access time to a particular memory address is independent regardless of which processor requests data from different memory banks. On the contrary, NUMA allows memory access time dependent on the relative processor location. That is, a processor has fast memory access time to its local memory and slow access to the rest of memory. This non-uniformity enables potential fewer memory accesses with fast access time. Limiting the number of memory accesses and fast memory accesses provides the key to high performance computing. Therefore, we decide to apply the NUMA architecture and propose a memory-centric run-time mapping to utilize the benefits of NUMA and reduce the amount of data communicated among processors.
\subsection{Memory-centric Run-time Mapping}
\indent At run-time, processes/clusters generated in Section 4 from each application are mapped onto processors in a NUMA architecture in order to be executed. Depending on the mapping (e.g., A process has to fetch data from the farthest memory bank.), data communication is a performance bottleneck. The goal is to improve the amortized time when slow accesses occur only once in a while and fast local accesses happen frequently. \\
\indent Therefore, without fully observing the structure of clusters with corresponding physical memories, performance would degrade due to these reasons: (1) Waiting for cache update: The multi-core platforms require the cache coherence protocol to have consistent data over private caches. A process later mapped to a different core may increase the time spent for the cache coherence protocol to fetch a cache line from the previous core. (2) Block memory operations between IOs and memory: Block memory operations in IOs constitute a large overhead in the program execution because a large amount of data are referenced and transferred between caches and main memory banks. (3) Core utilization: In an extreme case, processes may be mapped only onto one core to exploit cache temporal and spatial locality. However, the rest cores remain idle for a long time. Therefore, core utilization is another factor for efficient parallelism in multi-core systems.\\
\indent In order to improve performance, the run-time mapping should exploit and optimize the parallelism in clusters while considering data communication between clusters and resource utilization. Figure \ref{fig-mapping} shows three important factors to help design a memory-centric run-time mapping.\\
\begin{figure}
\centering
\includegraphics[width=0.73\textwidth, height=0.28\textwidth]{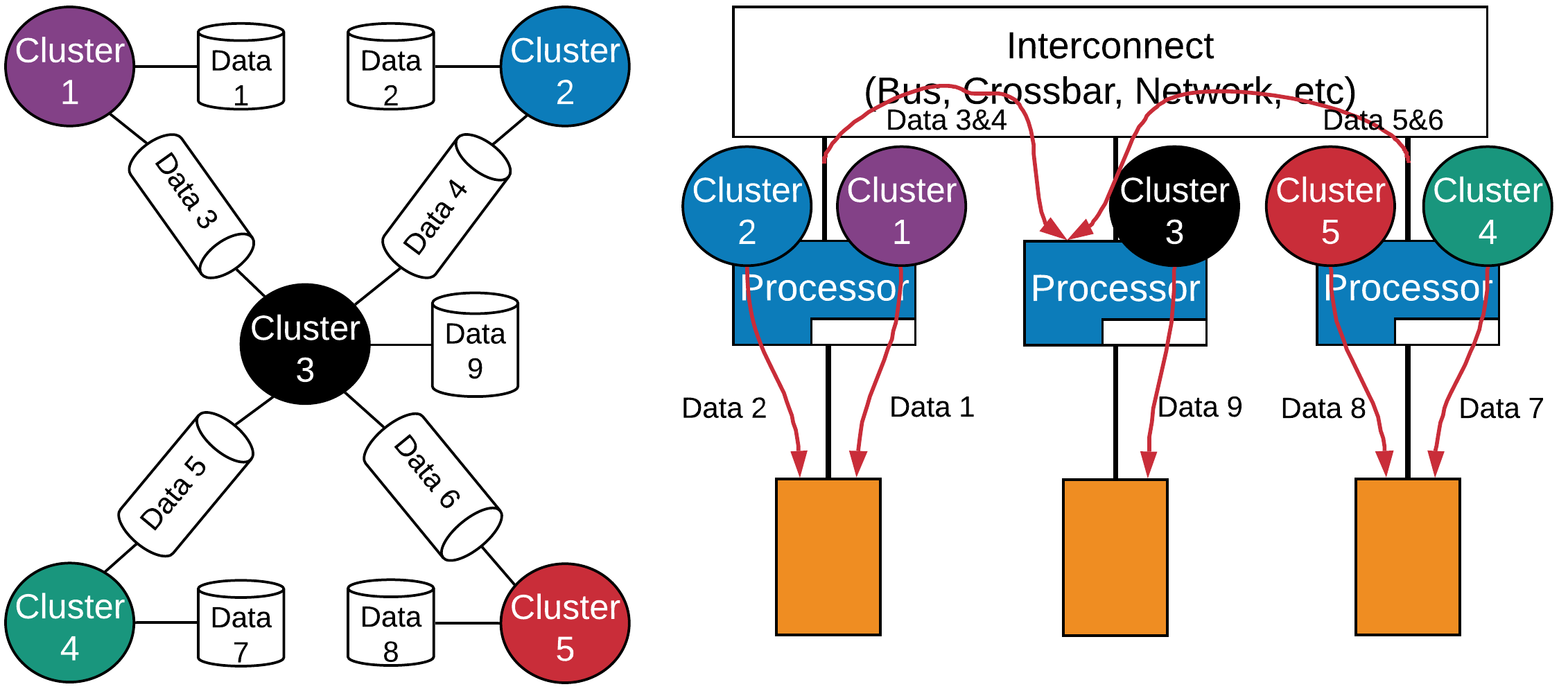} 
\caption{Memory-Centric Run-Time Mapping, which considers memory hierarchy and data communication.}
\vskip -10pt
\label{fig-mapping}
\end{figure}
\indent 1. Clusters which reference the same data structures can be mapped onto one core to prevent the time spent on cache coherence protocols and reduce the number of block operations. \\
\indent 2. Clusters which communicate with each other can be mapped to adjacent processors to improve the amortized time by reducing the number of times on fetching data from the farthest processor on a NUMA architecture.\\
\indent 3. Clusters which are independent of each other can be mapped to different regions of a multi-core platform (architecture decomposition) to reduce the number of sharing paths of messages. \\
\indent Therefore, the memory-centric run-time mapping algorithm as shown in Algorithm 2 takes as inputs the clusters, their interactions, and data communication and schedules a mapping from clusters to processors with the objective of improving application performance. The uttermost important criterion for a run-time mapping is the small time complexity. Therefore, we propose a greedy algorithm to achieve high performance with the time complexity of $O(P)$ where $P$ is the number of schedulable clusters. In the algorithm, we first check whether a cluster is ready to schedule, and keep track of all clusters with which this cluster communicates. Next, based on the factor 3, the architecture decomposition is performed to make sure that independent clusters are mapped onto faraway processors to distribute workloads and traffic evenly on hardware communication substrate in a multi-core platform. Then, we calculate the execution time for two clusters with the shared memory to be mapped onto the same processor and different ones, respectively. A mapping of the current cluster is decided based on factors 1 and 3. Mapping clusters with the shared memory onto the same processor could reduce the large overhead for block operations, but at the same time the parallelism may suffer if too many clusters are mapped to a single processor. Therefore, we set an upper threshold of the number of clusters to be mapped to a processor. In the evaluation, the threshold is 4.

\begin{algorithm}
\small
    \caption{Memory-Centric Run-Time Mapping}
  \begin{algorithmic}[1]
    \STATE \textbf{INPUT}: Clusters and data communication
    \STATE \textbf{OUTPUT}: A mapping from clusters to the architecture
    \STATE Runqueue RQ = $\emptyset$
    \FOR{$cluster$ in $clist$}
        \IF{$cluster$->$status$ = $SCHEDULABLE$}
            \IF{C = ClusterFromMem($cluster$->$m$) != $\emptyset$}
                \STATE $cluster$->$p$ = C // Factor 1
            \ENDIF
            \IF{C = ClusterComm($cluster$) != $\emptyset$}
                \STATE $cluster$->$ipc$ = C // Factor 2
            \ELSE
                \STATE $cluster$->$ipc$ = $\emptyset$ // Factor 3
            \ENDIF
            \STATE RQ.$push$($cluster$)
        \ENDIF
    \ENDFOR
    \STATE Regions = ArchitectureDecomposition()
    \STATE LastCluster = NULL
    \REPEAT
        \STATE $cluster$ = RQ.$pop$()
        \IF{ClusterFromMem($cluster$->$m$)->$core$ == LastCluster->$core$}
            \STATE // Decide on mapping clusters onto the same processor
            \IF{LastCluster->$core$->$num$ <= threshold}
                \STATE $cluster$->$core$ = $cluster$->$p$ // Factor 1
                \STATE LastCluster->$core$->$num$++
            \ELSE
                \STATE $cluster$->$core$ = DiffRegion($cluster$->$core$)
            \ENDIF
        \ELSIF{$cluster$->$ipc$ != $\emptyset$}
            \STATE $cluster$->$core$ = Nearby($cluster$->$ipc$->$core$)
        \ELSE
            \STATE $cluster$->$core$ = DiffRegion($cluster$->$ipc$->$core$)
        \ENDIF
        \STATE LastCluster = $cluster$
    \UNTIL{RQ is empty}
  \end{algorithmic}
\end{algorithm}
\section{Evaluations}
\indent In this section, we discuss the simulator configurations and present experimental results to investigate the soundness of the proposed methodology. 
\subsection{Simulation Configurations}
\begin{table}
\centering
\small
\caption{System Configuration}
\vskip -15pt
\scalebox{1}{
\begin{tabular}{ l|l|l }
\multicolumn{3}{}{} \\ 
\toprule
\multirow{6}{*}{\textbf{CPU}} & Cores & Out-of-order cores, 16 MSHRs \\ \cline{2-3}
 & Clock frequency & 2.4 GHz\\ \cline{2-3}
 & L1 private cache & 64KB, 4-way associative \\
 && 32-byte blocks \\ \cline{2-3}
 & L2 shared cache & 256KB, distributed \\ \cline{2-3}
 & Memory & 4 GB, 8 GB/s bandwidth \\ \hline
\multirow{3}{*}{\textbf{Network}} & Topology & Mesh \\ \cline{2-3}
 & Routing algorithm & XY routing \\ \cline{2-3}
 & Flow control & Virtual channel flit-based  \\ 
 \bottomrule
\end{tabular}
}
\label{table-simu}
\end{table}
\begin{table}[t]
\centering
\small
\caption{Summary and Description of Benchmarks 
}
\vspace{-7pt}
\scalebox{0.98}{
\begin{tabular}{l|l|l|l} \toprule
\textbf{Benchmark} & \textbf{Description} &\textbf{Input Size} & \textbf{Source} \\ \midrule
\textit{Dijkstra}&Find the shortest path&50 nodes&MiBench\\ \hline
\textit{FFT}&Compute fast Fourier transform&A vector of size 1024&OmpSCR\\ \hline
\textit{K-means}&Partition data into k clusters&128 2D tuples&Self-collected\\ \hline
\textit{Mandel}&Calculate Mandelbrot set&4092 points&OmpSCR\\ \hline
\textit{MD}&Simulate molecular dynamics&512 particles&OmpSCR\\ \hline
\textit{NN}&Neural networks&Three hidden fully connected layers&Self-collected\\ \hline
\textit{Neuron}&A list of neurons with the ReLU function&64 Neurons&Self-collected\\ \hline
\textit{CNN}&Convolutional neural networks&Conv-Pool-Conv-Pool-FC&Self-collected\\ \hline
\textit{Strassen8}&Strassen's matrix multiplication&Matrices of size 64&Self-collected\\ \hline
\textit{Strassen16}&Strassen's matrix multiplication&Matrices of size 256&Self-collected\\ \bottomrule
\end{tabular}
}
\vskip -5pt
\label{table-benchmark}
\end{table}
\indent We use gem5 \cite{binkert2011gem5} to simulate a varying number of out-of-order cores with the NUMA architecture. Table \ref{table-simu} shows detailed simulation parameters. Table \ref{table-benchmark} lists the considered applications from various benchmarks with different graph characteristics, including applications from the OpenMP Source Code Repository (OmpSCR)\cite{dorta2005openmp}, Mibench \cite{guthaus2001mibench}, and implementations of some other benchmark algorithms. We generate LLVM IR graphs from these applications following the procedure introduced in Section~\ref{sec:graph}. For baseline comparisons, we consider four baseline methods for graph partitioning: (1) the work in \cite{xiao2017load} abbreviated as \textit{CompNet}; (2) \textit{METIS} \cite{lasalle2015improving}, which is an edge-cut method that implements various multilevel algorithms by iteratively coarsening a graph, computing a partition, and projecting the partition back to the original graph; (3) the unweighted vertex-cut method \textit{PowerGraph} (\textit{PG}) \cite{gonzalez2012powergraph}; and (4) the unweighted vertex-cut method \textit{Libra} \cite{xie2014distributed}. We compare these baselines with the proposed four greedy vertex-cut algorithms: \textit{Weighted PowerGraph (W-PG)}, \textit{Weight Balanced PowerGraph (WB-PG)}, \textit{Weighted Libra (W-Libra)}, \textit{Weight Balanced Libra (WB-Libra)}. 

\subsection{Experimental Results}
In this section, we evaluate the proposed methods and baselines on the LLVM graphs transformed from the benchmarks listed in Table \ref{table-benchmark} for the proposed graph partitioning and compare their performance in the graph partition quality (in terms of replication factor and edge weight imbalance). Next, we execute clusters generated from each method to measure the overall execution time and data communication. We also analyze the sensitivity of the parameter $\lambda$ involved in the constraint of load balancing to the execution time.  
\subsubsection{Replication Factor.}
In Section~\ref{sec:replication}, we have derived the theoretical expected replication factor of the weighted vertex cut with random edge placement, which is in fact a theoretical upper-bound for the replication factor of the proposed greedy algorithms. We now empirically evaluate the performance of the proposed four greedy vertex cut algorithms in replication factor, and compare the results with the theoretical upper-bound we compute by Eq. (\ref{eq:expected_rep}). Fig.~\ref{fig:results_rep} shows the results on four graphs. As we can see, the four greedy algorithms achieve comparable performance in the replication factor. All of their replication factors are within the theoretical upper-bound with a considerable gap, which indicates the superior advantage of the greedy vertex-cut algorithms over the random vertex cut strategy. 
\subsubsection{Edge Weight Imbalance.}
As discussed in Section~\ref{sec:discussions}, edge weight imbalance is a key metric for evaluating the performance of vertex-cuts in load balancing among clusters. The edge weight imbalance is defined by $({\max_m\sum_{e\in E, M(e)=m} w_e})/{(\frac{w_{avg}|E|}{p}})$, which measures how much the most-loaded cluster deviates from the expected average load between clusters. A good load-balancing vertex-cut method should achieve an edge weight imbalance close to 1, which indicates the absolute balance. We evaluate the edge weight imbalance of all the six vertex cut methods, where we set $\lambda = 1$ in the sum of weights in a cluster (line 4 of Algorithm~\ref{algorithm:wblibra}) for WB-PG and WB-Libra, in order to obtain their optimal balance of edge weights for comparisons with the other methods. Table~\ref{tab:imbalance} shows the results from edge weight imbalance of the six methods on all ten graphs. We observe from the table that, WB-Libra achieves the best results in most of the graphs, except for Dijkstra, Mandel, and Md, where WB-PG performs the best. Both the two unweighted vertex cut methods, i.e. PowerGraph and Libra, achieve worse results compared to the four weighted vertex cut methods. This is mainly due to the fact that, the unweighted vertex cut was designed to balance the number of edges among clusters for unweighted graphs and therefore they can not properly handle the load balancing for weighted graphs. By comparing between WB-PG and W-PG, and between WB-Libra and W-Libra, we can see that the edge weight balance constraint we incorporate into the weighted balanced algorithms is effective for further improving the edge weight balance among clusters and is able to push the balance to the near-ideal situation.

\begin{figure}[h]
\centering
\begin{subfigure}[t]{0.42\textwidth}{
\centering
\includegraphics[width=1\textwidth,center]{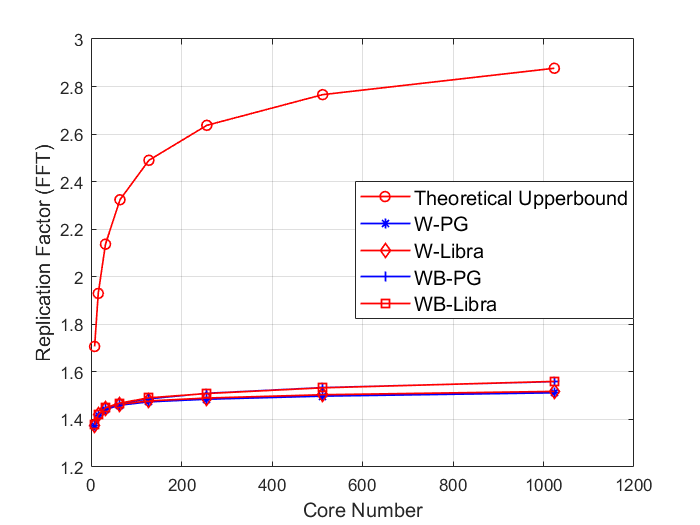}
}
\vspace{-9pt}
\caption{FFT}
\label{fig:fft_rep}
\end{subfigure}
\hspace{-8pt}
\begin{subfigure}[t]{0.42\textwidth}{
\centering

\includegraphics[width=1\textwidth,center]{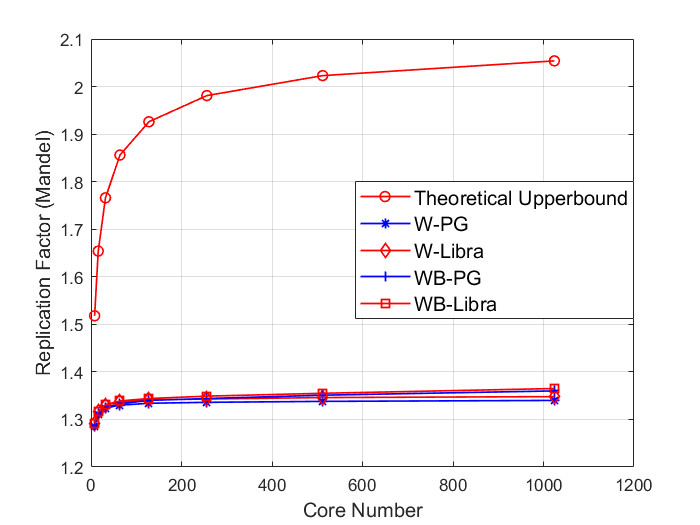}}

\caption{Mandel}
\label{fig:mandel_rep}
\end{subfigure}

\begin{subfigure}[t]{0.42\textwidth}{
\centering
\includegraphics[width=1\textwidth,center]{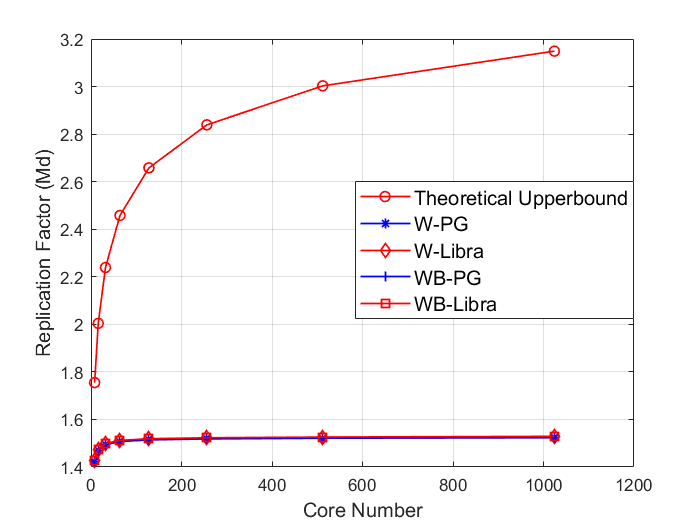}}
\caption{Md}
\label{fig:md_rep}
\end{subfigure}\hspace{-10pt}
\begin{subfigure}[t]{0.42\textwidth}{
\includegraphics[width=1\textwidth,center]{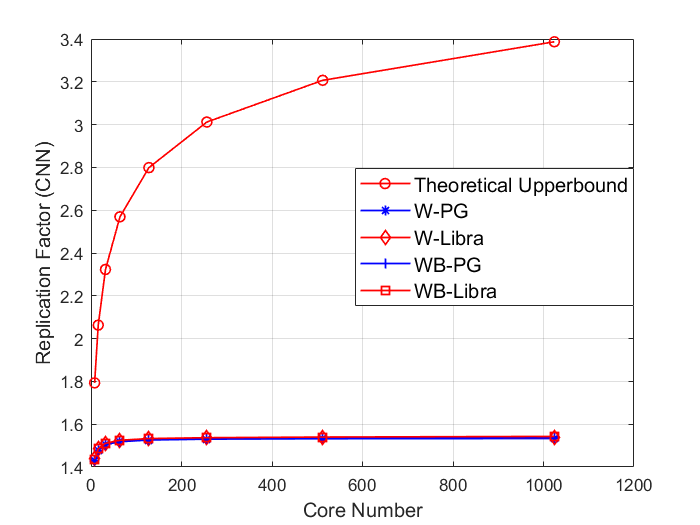}}
\caption{CNN}
\label{fig:cnn_rep}
\end{subfigure}
\vskip -5pt
\caption{Replication Factor of the Proposed Four Greedy Algorithms with Comparison to the Computed Theoretical Upper-bound by Eq. (\ref{eq:expected_rep})}
\label{fig:results_rep}
\end{figure}

\begin{table}[h]
\small
\caption{Statistics of Graph Datasets}
\vspace{-10pt}
\begin{tabular}{l|l|l|l|l}
\toprule
\textbf{Graph Dataset}        & \textbf{No. Nodes}   & \textbf{No. Edges}  & \textbf{power-law} $\mathbf{\alpha}$ & \textbf{Avg. Path Length}\\ \midrule
Dijkstra      & 248,959   & 291,112   & 2.29 & 136.4 \\ \hline
FFT           & 109,295   & 143,183   & 2.21 & 194.56\\ \hline
K-means       & 98,592    & 119,112   & 2.24 & 479.4 \\ \hline
Mandel        & 235,051   & 260,042   & 2.43 & 42.67 \\ \hline
Md            & 1,799,353 & 2,361,213 & 2.17 & 524.61\\ \hline
NN            & 124,496   & 161,428   & 2.16 & 171.52\\ \hline
Neuron        & 57,883    & 73,431    & 2.20 & 179.25\\ \hline
CNN           & 573,694   & 758,712   & 2.13 & 824.37\\ \hline
Strassen8     & 36,831    & 46,756    & 2.21 & 21.22\\ \hline
Strassen16    & 197,827   & 254,392   & 2.20 & 123.94\\ 
\bottomrule
\end{tabular}
\end{table}


\begin{table}[h]
\small
\caption{Edge Imbalance of the Vertex Cut Methods on Graphs} 
\label{tab:imbalance}
\vspace{-10pt}
\begin{tabular}{l|l|l|l|l|l|l}
\toprule
\textbf{Datasets}   & \textbf{PG} & \textbf{W-PG} & \textbf{WB-PG}    & \textbf{Libra}   & \textbf{W-Libra} & \textbf{WB-Libra}         \\ \hline
Dijkstra   & 1.00227    & 1.00092      & \textbf{1.00007} & 1.02136 & 1.00106 & 1.00010          \\ \hline
FFT        & 1.00586    & 1.00831      & 1.00075          & 1.05030 & 1.00400 & \textbf{1.00057} \\ \hline
K-means     & 1.00177    & 1.00469      & 1.00042          & 1.04566 & 1.00180 & \textbf{1.00035} \\ \hline
Mandel     & 1.00730    & 1.00233      & \textbf{1.00008} & 1.00749 & 1.00171 & 1.00014          \\ \hline
Md         & 1.00015    & 1.00007      & \textbf{1.00003} & 1.00791 & 1.00008 & \textbf{1.00003} \\ \hline
NN  & 1.00187    & 1.00235      & 1.00028          & 1.03441 & 1.00106 & \textbf{1.00019} \\ \hline
Neuron     & 1.00260    & 1.00487      & 1.00081          & 1.05738 & 1.00236 & \textbf{1.00045} \\ \hline
CNN  & 1.00040    & 1.00035      & 1.00010          & 1.00956 & 1.00027 & \textbf{1.00008} \\ \hline
Strassen8  & 1.01074    & 1.01307      & 1.00177          & 1.05036 & 1.00787 & \textbf{1.00123} \\ \hline
Strassen16 & 1.00338    & 1.00352      & 1.00029          & 1.04206 & 1.00170 & \textbf{1.00028} \\ 
\bottomrule
\end{tabular}
\end{table}

\subsubsection{Execution Time}
\begin{figure}
    \centering 
\begin{subfigure}{0.5\textwidth}
  \includegraphics[width=\linewidth, height=0.55\linewidth]{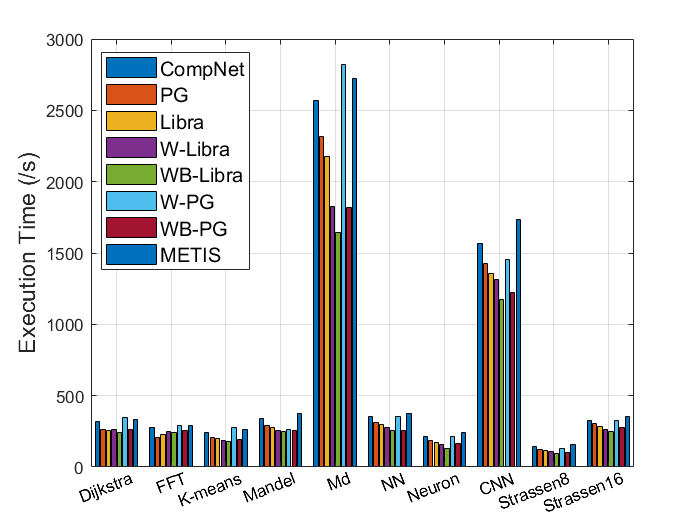}
  \vskip -10pt
  \caption{8 clusters}
  \label{fig:1}
\end{subfigure}\hfil 
\begin{subfigure}{0.5\textwidth}
  \includegraphics[width=\linewidth, height=0.55\linewidth]{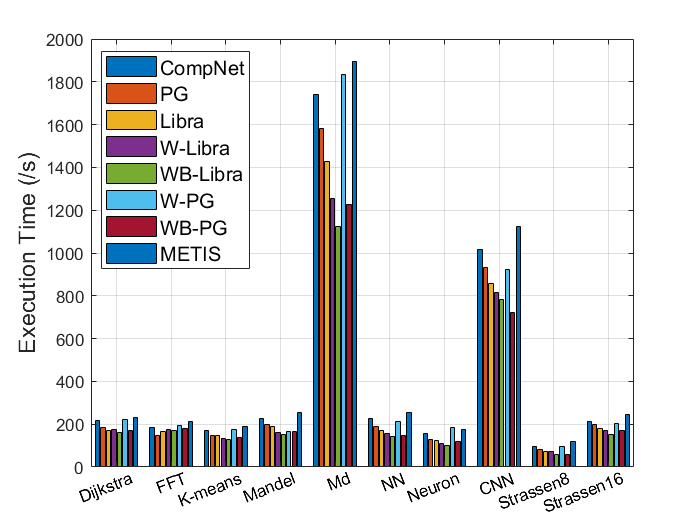}
 \vskip -10pt
  \caption{16 clusters}
  \label{fig:2}
\end{subfigure}\hfil 
\medskip
\begin{subfigure}{0.5\textwidth}
  \includegraphics[width=\linewidth, height=0.55\linewidth]{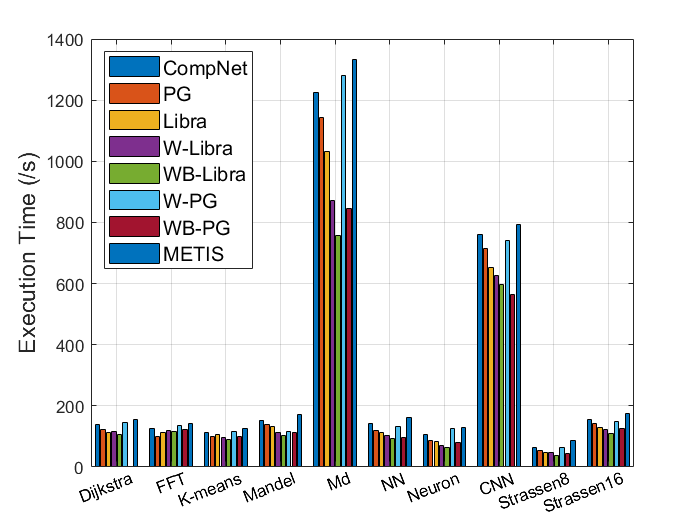}
   \vskip -10pt
  \caption{32 clusters}
  \label{fig:1}
\end{subfigure}\hfil 
\begin{subfigure}{0.5\textwidth}
  \includegraphics[width=\linewidth, height=0.55\linewidth]{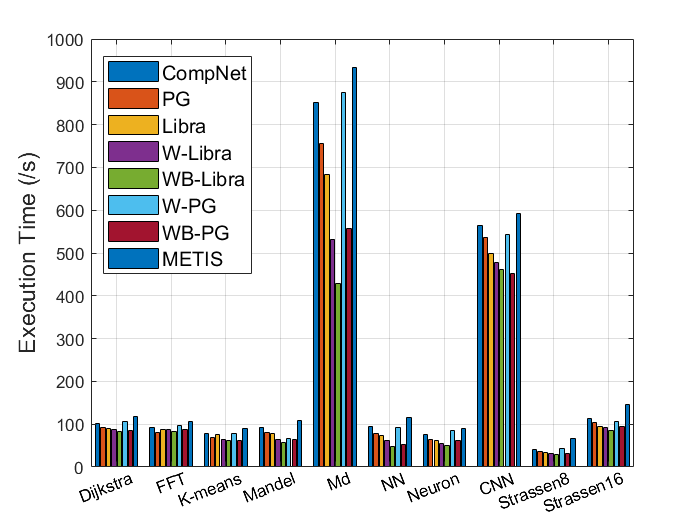}
   \vskip -10pt
  \caption{64 clusters}
  \label{fig:2}
\end{subfigure}\hfil 
\medskip
\begin{subfigure}{0.5\textwidth}
  \includegraphics[width=\linewidth, height=0.55\linewidth]{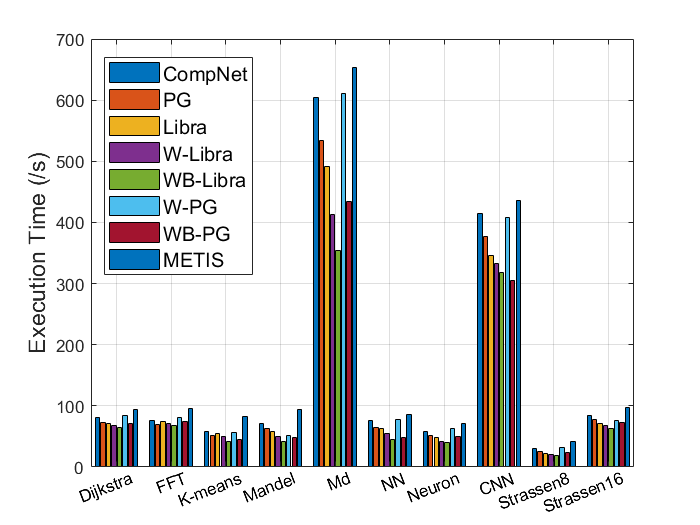}
   \vskip -10pt
  \caption{128 clusters}
  \label{fig:1}
\end{subfigure}\hfil 
\begin{subfigure}{0.5\textwidth}
  \includegraphics[width=\linewidth, height=0.55\linewidth]{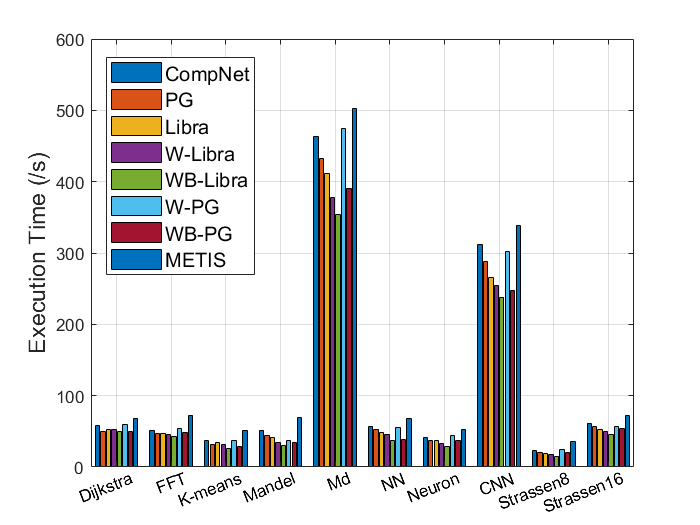}
   \vskip -10pt
  \caption{256 clusters}
  \label{fig:2}
\end{subfigure}\hfil 
\medskip
\begin{subfigure}{0.5\textwidth}
  \includegraphics[width=\linewidth, height=0.55\linewidth]{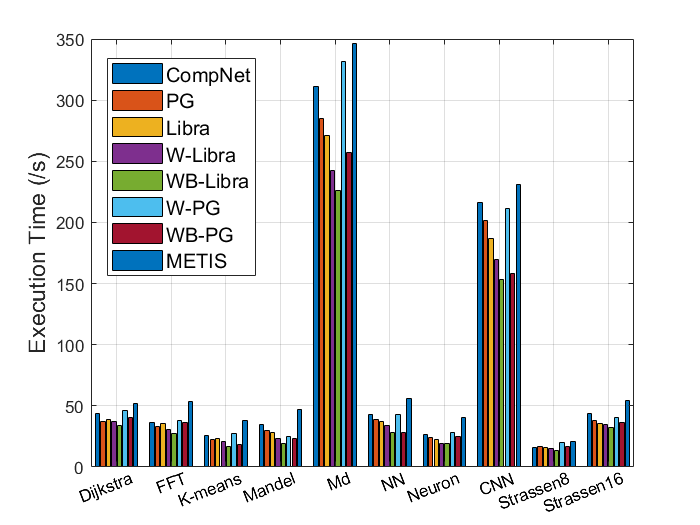}
   \vskip -10pt
  \caption{512 clusters}
  \label{fig:1}
\end{subfigure}\hfil 
\begin{subfigure}{0.5\textwidth}
  \includegraphics[width=\linewidth, height=0.55\linewidth]{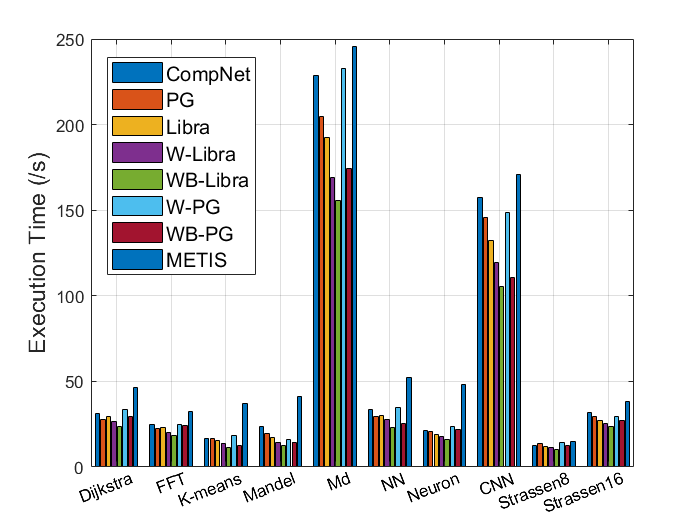}
   \vskip -10pt
  \caption{1024 clusters}
  \label{fig:2}
\end{subfigure}\hfil 
\vskip -10pt
\caption{Application Performance From Different Graph Partitioning Algorithms on a Multi-core System}
\label{fig-app}
\end{figure}
\begin{table}
\small
\caption{Overall Execution Time (/s) for 8 Clusters From Different Algorithms in a Multi-core Platform}
\vspace{-10pt}
\begin{tabular}{l|l|l|l|l|l|l|l|l}
 \toprule
\textbf{Datasets}   & \textbf{CompNet} & \textbf{METIS} & \textbf{PG}   & \textbf{W-PG} & \textbf{WB-PG} & \textbf{Libra} & \textbf{W-Libra} & \textbf{WB-Libra}  \\ \hline
Dijkstra   & 317.27 & 332.48 & 346.15     & 263.75          & 260.91        & 253.86  & 262.51  & \textbf{242.26}  \\ \hline
FFT        & 279.6         & 288.22     & \textbf{209.27} & 253.71        & 230.02  & 248.42  & 239.33  &  291.53        \\ \hline
K-means     & 244.87  & 261.38        & 279.53     & 206.25          & 195.53        & 201.54  & 188.25  & \textbf{178.58} \\ \hline
Mandel     & 341.35    & 373.28        & 265.15     & 289.74          & 257.12        & 277.31  & 256.49  & \textbf{245.82}\\ \hline
Md         & 2568.72  & 2723.71       & 2822.47    & 2313.9          & 1821.68       & 2178.41 & 1824.95 & \textbf{1642.18} \\ \hline
NN  & 351.23    & 376.93       & 354.32     & 311.74          & 256.41        & 297.59  & 278.24  & \textbf{253.89}\\ \hline
Neuron     & 214.75      & 242.68    & 213.95     & 187.23          & 163.63        & 174.54  & 157.69  & \textbf{131.4}  \\ \hline
CNN  & 1568.59   & 1736.37     & 1454.88    & 1425.63         & 1221.53       & 1358.61 & 1315.78 & \textbf{1175.8}  \\ \hline
Strassen8  & 142.41  & 155.39        & 131.24     & 121.37          & 104.99        & 112.62  & 111.23  & \textbf{96.24}  \\ \hline
Strassen16 & 326.75    & 351.26      & 323.5      & 304.21          & 274.63        & 285.44  & 264.58  & \textbf{248.25} \\ 
\bottomrule
\end{tabular}
\label{table:executiontime}
\end{table}
\begin{table}
\small
\caption{Overall Execution Time (/s) for 1024 Clusters From Different Algorithms in a Multi-core Platform}
\vspace{-10pt}
\begin{tabular}{l|l|l|l|l|l|l|l|l}
\toprule
\textbf{Datasets}   & \textbf{CompNet} & \textbf{METIS} & \textbf{PG}   & \textbf{W-PG} & \textbf{WB-PG} & \textbf{Libra} & \textbf{W-Libra} & \textbf{WB-Libra}  \\ \hline
Dijkstra   & 31.08 & 46.48 & 33.5     & 27.79          & 29.27        & 29.58  & 26.43  & \textbf{23.4}  \\ \hline
FFT        & 24.92         & 32.4     & 25.08 & 22.64        & 23.96  & 23.2  & 20.31  &  \textbf{18.59}        \\ \hline
K-means     & 16.77  & 37.23        & 18.26     & 16.37          & 12.54        & 15.3  & 13.92  & \textbf{11.26} \\ \hline
Mandel     & 23.48    & 41.37        & 15.8     & 19.81          & 14.52        & 17.47  & 14.43  & \textbf{12.6}\\ \hline
Md         & 228.43  & 245.61       & 233.02    & 204.53          & 174.23       & 192.23 & 169.18 & \textbf{155.71} \\ \hline
NN  & 33.44    & 52.36       & 34.92     & 29.35          & 25.48        & 29.84  & 27.9  & \textbf{23.22}\\ \hline
Neuron     & 21.3      & 48.32    & 23.73     & 20.93          & 21.62        & 19.19  & 17.97  & \textbf{16.11}  \\ \hline
CNN  & 157.5   & 170.92     & 148.48    & 145.87         & 110.67       & 132.22 & 119.43 & \textbf{105.29}  \\ \hline
Strassen8  & 12.51  & 15.03        & 14.57     & 13.55          & 12.39        & 12.17  & 11.11  & \textbf{10.31}  \\ \hline
Strassen16 & 31.81    & 38.14      & 29.23      & 29.62          & 27.36        & 27.22  & 25.25  & \textbf{23.42} \\ 
\bottomrule
\end{tabular}
\label{table:executiontime1024}
\end{table}
Fig.~\ref{fig-app} shows the execution time of each application for different graph partitioning algorithms with various cluster numbers. Specifically, tables \ref{table:executiontime} and \ref{table:executiontime1024} show the execution time for 8 and 1024 clusters, respectively. As we can see, the vertex-cut methods overall achieve a better performance the edge-cut baselines CompNet and METIS. This verifies our expectation that the vertex-cut based graph partitioning methods would work better than edge-cut methods for the power-law graphs. Among the six vertex-cut methods, the proposed four methods (i.e., W-PG, WB-PG, W-Libra and WB-Libra) outperform the two unweighted vertex-cut methods. This is reasonable since the unweighted vertex-cuts are not able to handle the load balancing for weighted graphs, as we discussed in Section~\ref{sec:discussions}, and the load imbalance among clusters will lead to a longer overall execution time for the applications, as the overall execution time depends on the time for executing the cluster with the largest workload. Among the four proposed methods, WB-Libra achieve the best performance in most cases consistently for all different numbers of clusters. This demonstrates the benefit of using the degree-based vertex hashing strategy and the load balancing constraint in the vertex-cuts. These results in execution time indicate that the proposed vertex-cut based graph partitioning framework is effective in load balancing and parallelism optimization for multi-core systems and it has superior performance than the state-of-the-art baselines. 

\subsubsection{Data Communication}
\begin{figure}
    \centering 
\begin{subfigure}{0.5\textwidth}
  \includegraphics[width=\linewidth, height=0.55\linewidth]{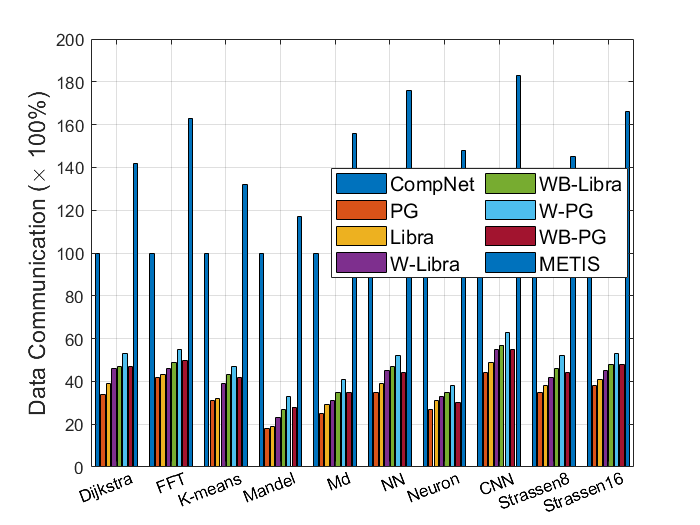}
  \vskip -10pt
  \caption{8 clusters}
  \label{fig:1}
\end{subfigure}\hfil 
\begin{subfigure}{0.5\textwidth}
  \includegraphics[width=\linewidth, height=0.55\linewidth]{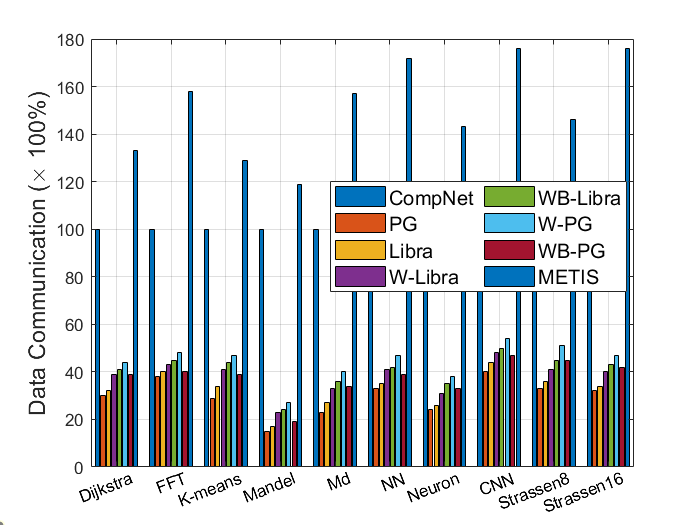}
 \vskip -10pt
  \caption{16 clusters}
  \label{fig:2}
\end{subfigure}\hfil 
\medskip
\begin{subfigure}{0.5\textwidth}
  \includegraphics[width=\linewidth, height=0.55\linewidth]{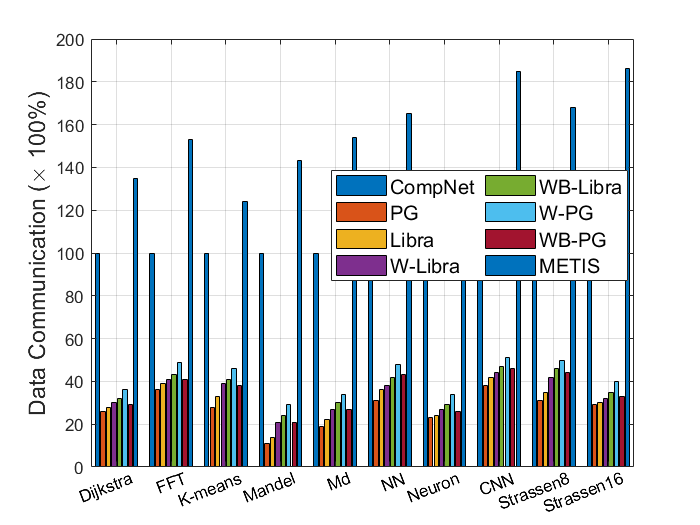}
   \vskip -10pt
  \caption{32 clusters}
  \label{fig:1}
\end{subfigure}\hfil 
\begin{subfigure}{0.5\textwidth}
  \includegraphics[width=\linewidth, height=0.55\linewidth]{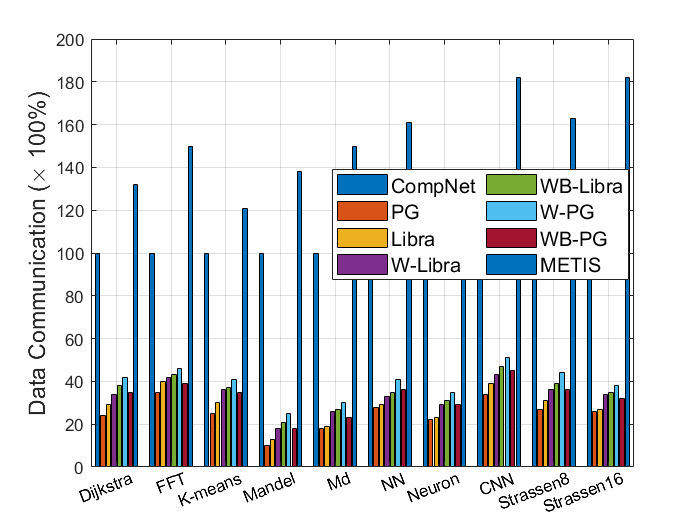}
   \vskip -10pt
  \caption{64 clusters}
  \label{fig:2}
\end{subfigure}\hfil 
\medskip
\begin{subfigure}{0.5\textwidth}
  \includegraphics[width=\linewidth, height=0.55\linewidth]{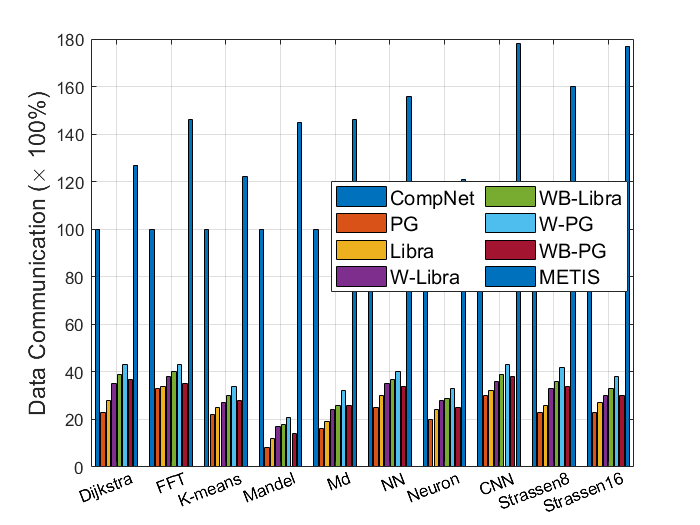}
   \vskip -10pt
  \caption{128 clusters}
  \label{fig:1}
\end{subfigure}\hfil 
\begin{subfigure}{0.5\textwidth}
  \includegraphics[width=\linewidth, height=0.55\linewidth]{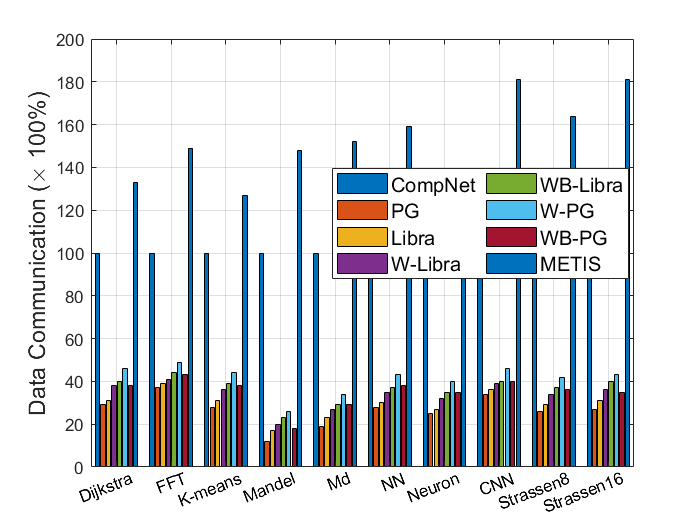}
   \vskip -10pt
  \caption{256 clusters}
  \label{fig:2}
\end{subfigure}\hfil 
\medskip
\begin{subfigure}{0.5\textwidth}
  \includegraphics[width=\linewidth, height=0.55\linewidth]{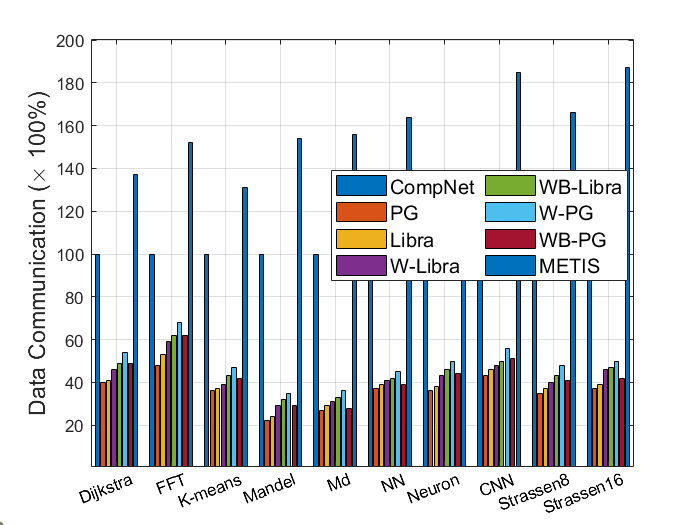}
   \vskip -10pt
  \caption{512 clusters}
  \label{fig:1}
\end{subfigure}\hfil 
\begin{subfigure}{0.5\textwidth}
  \includegraphics[width=\linewidth, height=0.55\linewidth]{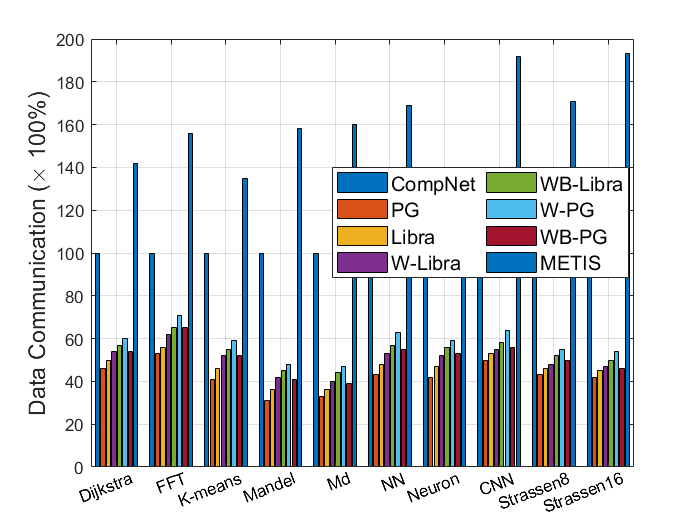}
   \vskip -10pt
  \caption{1024 clusters}
  \label{fig:2}
\end{subfigure}\hfil 
\vskip -5pt
\caption{Data Communication Cost From Different Graph Partitioning Algorithms in a Multi-core Platform}
\label{fig-datacomm}
\end{figure}
\begin{table}
\small
\caption{Data Communication for 8 Clusters From Different Graph Partitioning Algorithms}
\vspace{-10pt}
\begin{tabular}{l|l|l|l|l|l|l|l|l}
 \toprule
\textbf{Datasets}   & \textbf{CompNet} & \textbf{METIS} & \textbf{PG}   & \textbf{W-PG} & \textbf{WB-PG} & \textbf{Libra} & \textbf{W-Libra} & \textbf{WB-Libra}  \\ \hline
Dijkstra   & 100\%     &142\%     & 60\% & 46\% & 54\%  & 50\%  & 54\%    & 57\%    \\ \hline
FFT        & 100\%     & 156\%     & 71\% & 53\% & 65\%  & 56\%  & 62\%    & 65\%     \\ \hline
K-means     & 100\%     & 135\%       & 59\% & 41\% & 52\%  & 46\%  & 52\%    & 55\%   \\ \hline
Mandel     & 100\%    & 158\%      & 48\% & 31\% & 41\%  & 36\%  & 42\%    & 45\%     \\ \hline
Md         & 100\%    & 160\%      & 47\% & 33\% & 39\%  & 36\%  & 40\%    & 44\%     \\ \hline
NN  & 100\%    & 169\%      & 63\% & 43\% & 55\%  & 48\%  & 53\%    & 57\%    \\ \hline
Neuron     & 100\%  & 137\%    & 59\% & 42\% & 53\%  & 47\%  & 52\%    & 56\%     \\ \hline
CNN  & 100\%     & 192\%        & 64\% & 50\% & 56\%  & 53\%  & 55\%    & 58\%  \\ \hline
Strassen8  & 100\%   & 171\%        & 55\% & 43\% & 50\%  & 46\%  & 48\%    & 52\%    \\ \hline
Strassen16 & 100\%    & 193\%        & 54\% & 42\% & 46\%  & 45\%  & 47\%    & 50\%   \\ 
\bottomrule
\end{tabular}
\label{table-datacomm}
\end{table}
\begin{table}
\small
\caption{Data Communication for 1024 Clusters From Different Graph Partitioning Algorithms}
\vspace{-10pt}
\begin{tabular}{l|l|l|l|l|l|l|l|l}
 \toprule
\textbf{Datasets}   & \textbf{CompNet} & \textbf{METIS} & \textbf{PG}   & \textbf{W-PG} & \textbf{WB-PG} & \textbf{Libra} & \textbf{W-Libra} & \textbf{WB-Libra}  \\ \hline
Dijkstra   & 100\%     &142\%     & 53\% & 34\% & 47\%  & 39\%  & 46\%    & 47\%    \\ \hline
FFT        & 100\%     & 163\%     & 55\% & 42\% & 50\%  & 43\%  & 46\%    & 49\%     \\ \hline
K-means     & 100\%     & 132\%       & 47\% & 31\% & 42\%  & 32\%  & 39\%    & 43\%   \\ \hline
Mandel     & 100\%    & 117\%      & 33\% & 18\% & 28\%  & 19\%  & 23\%    & 27\%     \\ \hline
Md         & 100\%    & 156\%      & 41\% & 25\% & 35\%  & 29\%  & 31\%    & 35\%     \\ \hline
NN  & 100\%    & 176\%      & 52\% & 35\% & 44\%  & 39\%  & 45\%    & 47\%    \\ \hline
Neuron     & 100\%  & 148\%    & 38\% & 27\% & 30\%  & 31\%  & 33\%    & 35\%     \\ \hline
CNN  & 100\%     & 183\%        & 63\% & 44\% & 55\%  & 49\%  & 55\%    & 57\%  \\ \hline
Strassen8  & 100\%   & 145\%        & 52\% & 35\% & 44\%  & 38\%  & 42\%    & 46\%    \\ \hline
Strassen16 & 100\%    & 166\%        & 53\% & 38\% & 48\%  & 41\%  & 45\%    & 48\%   \\ 
\bottomrule
\end{tabular}
\label{table-datacomm1024}
\end{table}
Fig.~\ref{fig-datacomm} shows data communication of each application for different graph partitioning algorithms with various cluster numbers. Specifically, Tables \ref{table-datacomm} and \ref{table-datacomm1024} show the communication for 8 and 1024 clusters, respectively. As we can see, all the vertex cut methods have comparable good performance in reducing data communication and outperform the edge-cut methods (METIS and CompNet) by a huge amount. For example, according to Table~\ref{table-datacomm}, the WB-PG reduces the data communication for 8-cores by an average of $48.9\%$ compared to CompNet over the 10 graphs, and WB-Libra reduces it by an average of $46.1\%$. METIS fails to reduce data communication compared to others. However, it is interesting to note that METIS causes less than 120\% for the Mandelbrot application whereas the data communication is more than 130\% for the rest of applications. It is because Mandelbrot is a embarrassingly parallel application where little effort is required to separate it into a number of parallel clusters. However, traditional edge cut algorithms such as CompNet and METIS still lead to a seriously huge amount of data communication between clusters, while vertex cut methods is able to maintain a much lower communication cost. This is mainly because that the data communication in edge-cut partitions comes from all the inter-cluster edges, while there is no such communication cost in vertex-cut partitions since there is no inter-cluster edges and the only communication for the vertex-cut partitions is the communication between the replicas of the cut vertices across different clusters. Another thing to notice is that the general trend of data communication from 8 clusters up to 1024 clusters is it first goes down and up again at 128 clusters. The trend of data communication going down is mainly due to the efficient parallelism while minimizing data communication. However, as the number of clusters increases beyond 128 clusters, synchronization starts to take over the impact of data communication because processes are synchronized to allow only one process enter the critical section to modify the shared data structures in main memory. Nevertheless, the least data communication overhead in these cases is still from the proposed vertex-cut algorithms.

\subsubsection{Sensitivity Analysis: Execution Time v.s. Parameter $\lambda$}
\begin{figure}
    \centering 
\begin{subfigure}{0.42\textwidth}
  \includegraphics[width=\linewidth, height=0.55\linewidth]{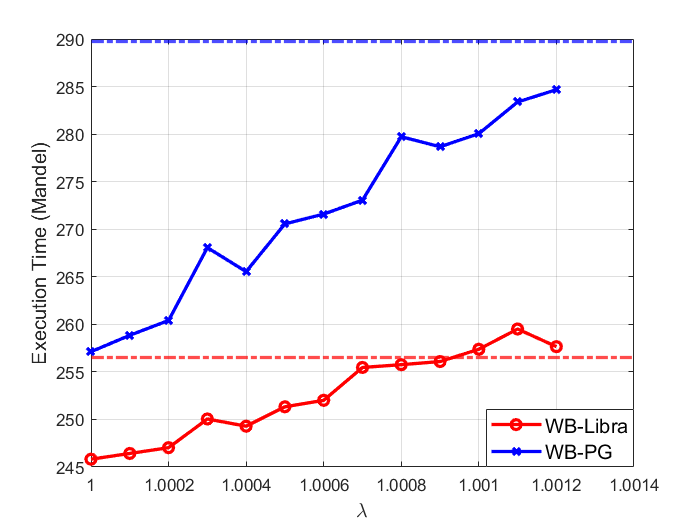}
  \vskip -5pt
  \caption{Mandel}
  \label{fig:1}
\end{subfigure}\hfil 
\begin{subfigure}{0.42\textwidth}
  \includegraphics[width=\linewidth, height=0.55\linewidth]{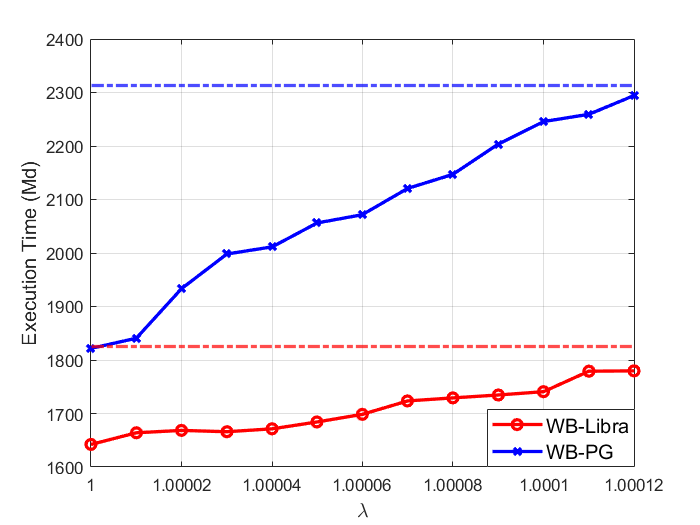}
 \vskip -5pt
  \caption{Md}
  \label{fig:2}
\end{subfigure}\hfil 
\medskip
\begin{subfigure}{0.33\textwidth}
  \includegraphics[width=1.\linewidth, height=0.8\linewidth]{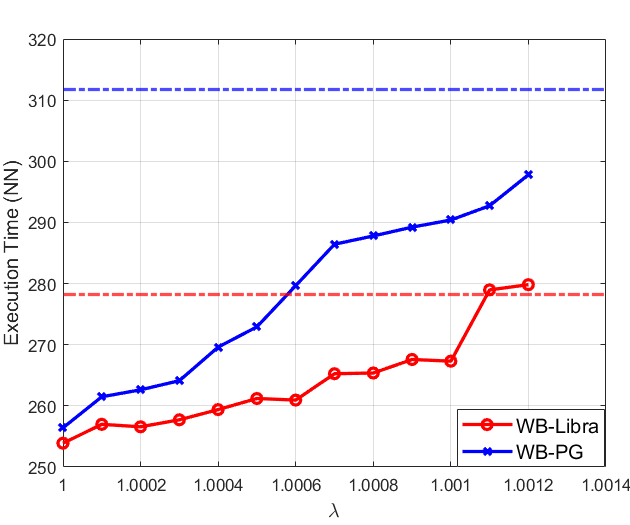}
   \vskip -5pt
  \caption{NN}
  \label{fig:1}
\end{subfigure}\hfil 
\begin{subfigure}{0.33\textwidth}
  \includegraphics[width=1.\linewidth, height=0.8\linewidth]{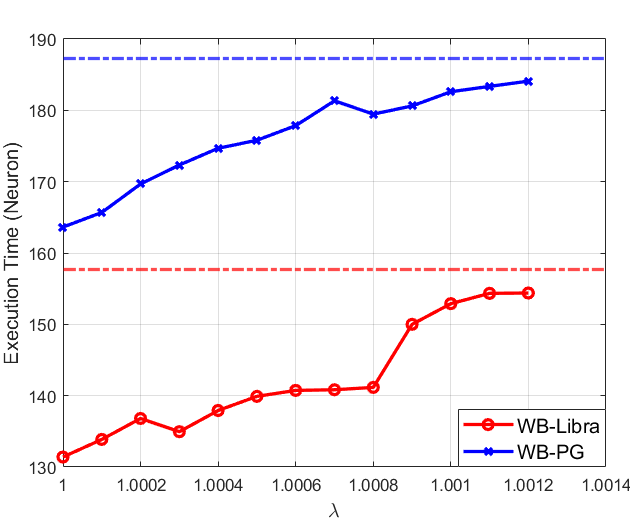}
   \vskip -5pt
  \caption{Neuron}
  \label{fig:2}
\end{subfigure}\hfil 
\begin{subfigure}{0.33\textwidth}
  \includegraphics[width=1.0\linewidth, height=0.8\linewidth]{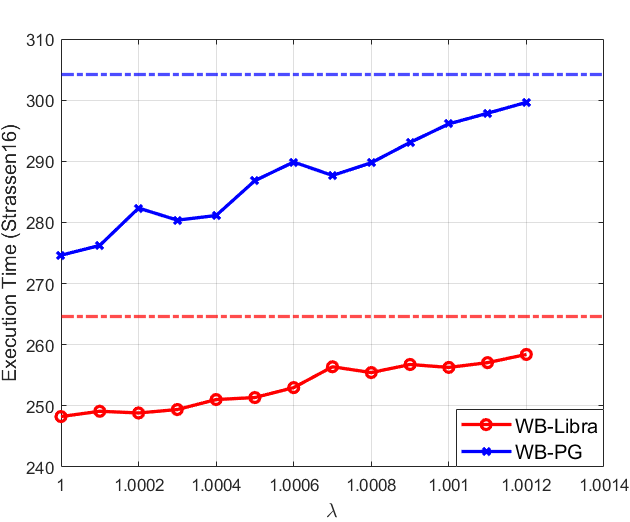}
   \vskip -5pt
  \caption{Strassen16}
  \label{fig:1}
\end{subfigure}\hfil 
\vskip -10pt
\caption{Execution Time With Different $\lambda$ Values for WB-Libra and WB-PG Algorithms. Dotted lines indicate the execution time for W-Libra and W-PG to which WB-Libra and WB-PG reduce, respectively, when $\lambda$ becomes large enough.}
\label{fig-par}
\end{figure}

Fig.~\ref{fig-par} shows the execution time with different $\lambda$ values in Eq. (\ref{eq:constraint}) on five applications, i.e., Mandel, Md, NN, Neuron, and Strassen16. $\lambda$ controls the balance of the clusters. In the WB-PG and WB-Libra algorithms, we explicitly use $\lambda$ to set a balance bound (line \#4 in Algorithm 1), and $\lambda = 1$ indicates the ideal balanced case. When $\lambda$ is large enough, WB-Libra and WB-PG reduce to W-Libra and W-PG, respectively. To analyze the sensitivity of the parameter $\lambda$ to the execution time, we evaluate WB-Libra and WB-PG with different $\lambda$ values in the range of 1 to 1.00012 with a step size of 0.00001 for Md due to its large size and 1 to 1.0012 with a step size of 0.0001 for the rest graphs. The dotted blue line refers to the performance of W-PG and the dotted red line refers to the performance of W-Libra, which can be treated as upper bound for WB-PG and WB-Libra, respectively. There are times when the execution time of applications exceeds the upper bound indicated by the dotted lines because of frequent synchronization such as fetching data from memory and flushing dirty data into memory. In general, the trend is going up, indicating that increasing $\lambda$ causes the performance degradation. It is recommended to set $\lambda=1$ in WB-Libra and WB-PG to improve the execution time.

\section{Related Work}
Parallel computing enables the continued growth of complex applications \cite{asanovic2006landscape, asanovic2009view}. Most existing work \cite{murray2013naiad, fetterly2009dryadlinq, murray2011ciel} exploits the coarse-grained parallelism of the dataflow graphs where it is common to represent computations as nodes and data dependencies among them as edges. The work in~\cite{murray2013naiad} proposes a new computational model, timely dataflow, and captures opportunities for parallelism across different algorithms. Timely dataflow combines dataflow graphs with timestamps to allow vertices send and receive logically timestamped messages along directed edges for data parallel computation in a distributed cluster. \cite{fetterly2009dryadlinq} proposes DryadLINQ, a system for general-purpose data-parallel computation. The system architecture incorporates the dataflow graph representation of jobs with a centralized job manager to schedule jobs on clustered computers. \cite{murray2011ciel} introduces CIEL, a universal execution engine for distributed dataflow programs. It coordinates the distributed execution of a set of data-parallel tasks, and dynamically builds the DAG as tasks execute. Others develop different edge-cut graph partitioning algorithms in parallel computing such as spectral graph theory \cite{hendrickson1995improved}, hypergraph models \cite{hendrickson2000graph,devine2006parallel}, and a multi-level algorithm \cite{hendrickson1995multi}. Few \cite{xiao2017load,xiao2018prometheus} exploit the fine-grained instruction parallelism in high-level programs and propose community detection inspired optimization models to benefit from the underlying hardware such as multi-core platforms and processing-in-memory architectures.\\
\indent Related works in vertex cut are mainly from the distributed graph computing field, where vertex-cuts are used to partition large power-law graphs for optimizing the distributed execution of real applications such as \textit{PageRank}. The PowerGraph \cite{gonzalez2012powergraph} and Libra \cite{xie2014distributed} discussed in previous sections are two state-of-the-art works in this field. Some other relevant works include \cite{jain2013graphbuilder}, \cite{gonzalez2014graphx}, and \cite{chen2019powerlyra}.

\section{Conclusion}
In this paper, we explore IR instruction-level parallelism via graph partitioning on LLVM IR graphs and cluster-to-core mapping for optimizing the parallel execution of applications on multi-core systems. we propose a vertex cut-based framework on LLVM IR graphs for load balancing and parallel optimization of application execution on multi-core systems. Specifically, we formulate a new problem called Weight Balanced $p$-way Vertex Cut by incorporating the weights into the optimization of vertex-cut graph partitioning, and we provide greedy algorithms for solving this problem. We also introduce a memory-centric run-time mapping algorithm for mapping the graph clusters to the multi-core architecture. Our simulation results demonstrate the superior performance of the proposed framework for load balancing and multi-core execution speed-up compared to the state-of-the-art baselines.





\end{document}